\documentclass[twocolumn,showpacs,amsmath,amssymb,pra]{revtex4}
\usepackage{theorem}
\usepackage{graphicx}
\theorembodyfont{\rmfamily}   
\newtheorem{theorem}{Theorem}
\newtheorem{lemma}[theorem]{Lemma}
\newtheorem{definition}[theorem]{Definition}
\newtheorem{corollary}[theorem]{Corollary}

\newtheorem{remark}[theorem]{Remark}
\newcommand{\qed}{\hfill$\square$}

\newenvironment{proof}{%
  \noindent{\em Proof.\ }}{%
  \hspace*{\fill}\qed \\
  \vspace{2ex}}
\DeclareMathAlphabet{\bm}{OML}{cmm}{b}{it}
\newcommand{\ket}[1]{| #1 \rangle} 
\newcommand{\bra}[1]{\langle #1 |} 

\newcommand{\rom}[1]{\mathrm{#1}}
\newcommand{\san}[1]{\mathsf{#1}}
\newcommand{\mbf}[1]{\mathbf{#1}}

\begin{document}

\title{Key rate of quantum key distribution
 with hashed two-way classical 
communication\footnote{Part of this paper will be presented at the 
2007 IEEE International Symposium on Information Theory, 
Acropolis Congress and Exhibition Center, Nice, France, 24th--29th June
 2007,
and will be published in its proceedings without proofs.}}

\author{Shun Watanabe}
 \email{shun-wata@it.ss.titech.ac.jp}
\author{Ryutaroh Matsumoto}%
 \email{ryutaroh@rmatsumoto.org}
 \homepage{http://www.rmatsumoto.org/research.html}
\author{Tomohiko Uyematsu}
 \email{uyematsu@ieee.org}
 \affiliation{%
Department of Communicaiton and Integrated Systems \\
Tokyo Institute of Technology \\
2-12-1, Oookayama, Meguro-ku, Tokyo, 152-8552, Japan
}

\author{Yasuhito Kawano}%
 \email{kawano@theory.brl.ntt.co.jp}
\affiliation{%
NTT Communication Science Laboratories, \\ 
NTT Corporation \\
3-1, Wakamiya, Morinosato, Atsugishi, \\
Kanagawa Pref., 243-0198, Japan
}%


\begin{abstract}
We propose an information reconciliation
protocol that uses two-way classical communication.
The key rates of quantum key distribution (QKD) protocols 
that use our new
protocol are higher than those of previously known protocols for a wide
range of error rates 
for the BB84 and six-state protocols.
We also clarify the relation between
the proposed  and known QKD protocols, and the relation 
between it and 
entanglement distillation protocols (EDPs). 
\end{abstract}

\maketitle

\section{Introduction}         

Quantum key distribution (QKD) protocols provide a way for
two parties, a sender,  Alice, and a receiver, Bob, to share an
unconditionally  secure key
in the presence of an eavesdropper, Eve. Unlike conventional
schemes of key distribution that rely on unproven computational
assumptions,  the security of QKD protocols is guaranteed by the 
principles of quantum mechanics.

QKD protocols usually consist of two parts, a quantum  and
a classical part.
Alice sends a binary sequence to Bob in the quantum part by
encoding it into quantum states that are randomly chosen from
a set of non-orthogonal states.
Since unknown non-orthogonal states cannot be cloned perfectly, any
eavesdropping attempt by Eve will disturb the transmitted quantum states.
Thus, by estimating the error rate of the transmitted quantum states, Alice and
Bob can estimate the amount of information that Eve has gained. 
For the sequence that remains after the error estimation phase,
which is usually called the raw key, 
Alice and Bob first carry out an information reconciliation 
(IR) protocol \cite{brassard:94}
to share the same bit sequence.
Alice and Bob then distill the final secure key by 
conducting a privacy
amplification (PA) protocol \cite{bennett:95}.

The best-known QKD protocols are the Bennett-Brassard 1984 (BB84) protocol \cite{bennett:84}
and the six-state protocol \cite{bruss:98}.
The unconditional security of the BB84 protocol has been proved 
\cite{biham:00, biham:06, mayers:01}.
Shor and Preskill \cite{shor:00} presented a simple proof of the BB84 protocol
by showing that the QKD protocol that uses the entanglement distillation
protocol (EDP) \cite{bennett:96, bennett:96b} can be converted into 
the BB84 protocol.
After that, the unconditional security of the six-state protocol was 
proved \cite{lo:01} by using the same technique as Shor and Preskill
used \cite{shor:00}.
Recently, the security of generic QKD protocols that include the BB84 protocol
and the six-state protocol has been proved \cite{kraus:05, renner:05, renner:05b},
which are based on
information theoretical techniques instead of 
Shor and Preskill's technique. 

In addition to the security of QKD protocols, the key rates of QKD protocols are also important,
where the key rate is defined by the ratio of the length of the final secure key to
the length of the raw key.
Gottesman and Lo \cite{gottesman:03} converted 
EDPs that use two-way classical communication 
into QKD protocols
that use the same communication.
More specifically, they proposed preprocessing that uses two-way classical
communication.
By inserting this two-way preprocessing before the conventional one-way
IR protocol, the key rates
of QKD protocols are increased when the error rate of a channel 
expressed as a percentage is more than about 9 \%.
Indeed, the tolerable error rate of the BB84 protocol is increased from 11 \% to 18.9 \%,
and that of the six-state protocol is increased from 12.7 \% to 26.4 \%,
where the tolerable
error rate is the error rate at which the key rate becomes zero.
Chau later showed that the two-way BB84 protocol can tolerate 20.0 \%
error rate,
 and that the 
two-way six-state protocol can tolerate 27.6 \% error rate \cite{chau:02}.
Recently, this kind of two-way preprocessing has been applied to 
QKD protocols with weak coherent pulses \cite{ma:06, kraus:07}.
It should be noted that this preprocessing is also known within the classical
key agreement context, in which it is usually called an advantage
distillation protocol \cite{maurer:93}.
Bae and Ac\'in and Ac\'in et al. \cite{acin:06, bae:07} extensively studied
the noise tolerance of QKD protocols with
advantage distillation protocols, on the other hand,
we are interested in the key rates of QKD protocols in this paper.

Vollbrecht and Vestraete proposed
a new type of two-way EDP \cite{vollbrecht:05}. 
This protocol uses
previously shared EPR pairs as an assistant resource (two-way breeding EDP), and the distillation rate of
this EDP exceeds that of one-way EDPs for a whole range of 
fidelities,
where a fidelity is that between the initial mixed
state and the EPR pair.
Using
the fact that a breeding EDP can be converted into a QKD 
protocol assisted by one-time pad encryption 
with a pre-shared secret key \cite{lo:03}, 
Vollbrecht and Vestraete's two-way breeding EDP \cite{vollbrecht:05} was
converted into a two-way QKD protocol assisted by one-time pad encryption \cite{ma:06, watanabe:06}.
The key rate of the converted QKD protocol is higher than that of
one-way QKD protocols \cite{shor:00, lo:01}
for a whole range of error rates.
It should be noted that the use of a pre-shared secret key is not 
the basis of their improvement,
because any QKD protocol that makes use of a pre-shared key
can be transformed into an equally efficient protocol that 
does not need a pre-shared
secret key \cite{christandl:06}.

We propose an IR protocol that uses
two-way classical communication in this paper.
Our proposed protocol is based on 
Vollbrecht and Vestraete's idea of two-way breeding
EDP \cite{vollbrecht:05}, but does not 
require any pre-shared secret keys.
Furthermore, 
our protocol does not leak information that is
redundantly leaked to Eve \cite{ma:06, watanabe:06}.
More precisely, in these protocols \cite{ma:06, watanabe:06},
Alice sends a redundant message that is useless to Bob, but
is useful to Eve.
However, in the proposed protocol, Alice does not
send that redundant information.
As a result, for the BB84 and six-state protocol,
the key rates of the QKD protocols that use our IR
protocol are higher than those of previously known
protocols for a wide range of error rates.
Especially, the key rate of our protocol is higher than
those of known protocols \cite{watanabe:06, shor:00, lo:01, renner:05} for the whole
range of error rates.
We also show the relation between the proposed
protocol and the advantage distillation protocol, i.e.,
the B-step of Gottesman and Lo \cite{gottesman:03}
(Remark \ref{remark-relation-to-b-step}).
We also show the relation between the proposed QKD protocol and
Vollbrecht and Vestraete's EDP. As a results, it turns out that
there does not seem to be any EDP that corresponds to
our proposed protocol (Remark \ref{remark-relation-to-vollbrecht}).

The rest of this paper is organized as follows.
Section \ref{two-way-information-reconciliation-protocol}
proposes a two-way IR protocol.
Section \ref{security-of-qkd-and-key-rate}
presents the key rate formula of the QKD protocol that uses our
proposed IR protocol.
There is a proof of the key rate formula in
the Appendix \ref{sec-proof-of-theorem}.
Section \ref{annalysis-of-key-rate}
presents the key rate formula as a function of error rate.
The proof of this formula is presented in Appendix \ref{proof-of-theorem-3}.

\section{Two-way information reconciliation protocol}
\label{two-way-information-reconciliation-protocol}

We propose an 
IR protocol that uses two-way classical communication
(called two-way IR protocol after this) in this section.
When Alice and Bob have correlated classical sequences,
$\mbf{x}, \mbf{y} \in \mathbb{F}_2^{2n}$, the purpose of  IR protocols
for Alice and Bob is to share the same classical sequence by
exchanging messages over a public authenticated channel, where
$\mathbb{F}_2$ is the field of order $2$.
Here, we assume that the pair of sequences 
$(\mbf{x}, \mbf{y})$ is independently identically distributed (i.i.d)
according to a joint probability distribution, $P_{XY}$, on 
$\mathbb{F}_2 \times \mathbb{F}_2$. 

Let us  review some notations 
for a linear code to describe our IR protocol.
An $[n, n-m]$ linear code, ${\cal C}_{n,m}$, is an 
$(n-m)$-dimensional linear subspace of $\mathbb{F}_2^n$.
Then, a parity check matrix, $M_{{\cal C}_{n,m}}$, of code ${\cal C}_{n,m}$ is an
$m \times n$ matrix of rank $m$ with $0, 1$ entries such that
$\mbf{c} M_{{\cal C}_{n,m}}^T = \mbf{0}$ for any $\mbf{c} \in {\cal C}_{n,m}$,
where $M_{{\cal C}_{n,m}}^T$ is the transpose matrix of $M_{{\cal C}_{n,m}}$.
A decoder, $g_{{\cal C}_{n,m}}$,
of code ${\cal C}_{n,m}$ is a map
from a syndrome, $\mbf{t} \in \mathbb{F}_2^m$, to an error, 
$\mbf{e} \in {\cal D}(\mbf{t})$,
where ${\cal D}(\mbf{t}) := \{ \mbf{e} \in \mathbb{F}_2^n \mid \mbf{e} M_{{\cal
C}_{n,m}}^T =  \mbf{t} \}$ is the  set of errors whose  syndromes are $\mbf{t}$.
After this, we will assume that a linear code is implicitly
specified with a parity check matrix  
and a decoder.

We need to define
some auxiliary random variables to describe our IR protocol.
Let $\xi_1:\mathbb{F}_2^2 \to \mathbb{F}_2$ be a function defined as
$\xi_1(a_1,a_2) := a_1 + a_2$ for $a_1,a_2 \in \mathbb{F}_2$, and
let $\xi_2:\mathbb{F}_2^2 \to \mathbb{F}_2$ be a function defined as
$\xi_2(a,0) := a$ and $\xi_2(a,1):=0$ for $a \in \mathbb{F}_2$.
For a pair of joint random variables $((X_1, Y_1)$, $(X_2, Y_2))$ with a
distribution, $P_{XY}^2$, define random variables
$U_1 := \xi_1(X_1, X_2)$, $V_1:= \xi_1(Y_1, Y_2)$ and
$W_1:= U_1 + V_1$.
Furthermore, define random variables  $U_2 := \xi_2(X_2,W_1)$,
$V_2:= \xi_2(Y_2,W_1)$ and $W_2 := U_2 + V_2$.
For the pair of sequences, 
$\mbf{x} = (x_{11},x_{12},\ldots,x_{n1},x_{n2})$ and
$\mbf{y} = (y_{11},y_{12}, \ldots, y_{n1}, y_{n2})$,
which is distributed according to 
the product distribution, $P_{XY}^{2n}$,
let $\mbf{u}$, $\mbf{v}$ and $\mbf{w}$ be $2n$-bit sequences such that
\begin{eqnarray*}
u_{i1} := \xi_1(x_{i1}, x_{i2}),~~
v_{i1} := \xi_1(y_{i1}, y_{i2}),~~
w_{i1} := u_{i1} + v_{i1}
\end{eqnarray*}
and
\begin{eqnarray*}
u_{i2} := \xi_2(x_{i2}, w_{i1}),~~
v_{i2} := \xi_2(y_{i2}, w_{i1}),~~
w_{i2} := u_{i2} + v_{i2}
\end{eqnarray*}
for $1 \le i \le n$.
Then, the pair $(\mbf{u},\mbf{v})$ is distributed
according to the distribution, $P_{U_1 U_2 V_1 V_2}^n$,
and the discrepancy, $\mbf{w}$, between $\mbf{u}$ and
$\mbf{v}$ is distributed according to the distribution,
$P_{W_1 W_2}^n$.
For sequence $\mbf{w}$, let 
$\san{T}_b := \{ j \mid 1 \le j \le n,~ w_{j1} = b\}$
be the set of indices of blocks such that the parities
of the discrepancies are $b$.
For the subsequence, $\mbf{u}_2 := (u_{12}, \ldots, u_{n2})$,
let $\mbf{u}_{2, \san{T}_b}$ be the subsequence
that consists of the $i$-th bit of $\mbf{u}_2$
such that $i \in \san{T}_b$.

Well-known 
methods \cite{gottesman:03, maurer:93, vollbrecht:05}
of two-way processing within the key distillation context
have been to classify blocks of length $2$ according to the parity, $w_{i1}$, 
of the discrepancies in each block.
In conventional two-way processing of the key distillation protocols
\cite{gottesman:03, maurer:93}, which is so-called advantage
distillation protocols, Alice sends the parity sequence,
$\mbf{u}_1 := (u_{11},\ldots,u_{n1})$, to Bob so that he can 
identify the parity sequence, 
$\mbf{w}_1 := (w_{11},\ldots, w_{n1})$, of the discrepancies.
Then, Alice and Bob discard $\mbf{u}_1$ and 
$\mbf{v}_1 := (v_{11},\ldots,v_{n1})$ respectively, because
$\mbf{u}_1$ is revealed to Eve. 
Furthermore, Alice and Bob discard
the second bit of the $i$-th block, if
the parity of the discrepancies is $1$, i.e.,
$i \in \san{T}_1$. 
Finally, Alice and Bob undertake an error correction procedure
for the subsequences $(\mbf{u}_{2,\san{T}_0}, \mbf{v}_{2, \san{T}_0})$.
More precisely, Alice sends the syndrome, 
$\mbf{t}_2 := \mbf{u}_{2,\san{T}_0} M_{{\cal C}_{n_0, m_0}}^T$,
for the prescribed $[n_0, m_0]$-linear code, and then
Bob decodes $\mbf{\hat{w}}_{2, \san{T}_0} := 
g_{{\cal C}_{n_0, m_0}}( \mbf{t}_2 + \mbf{v}_{2,\san{T}_0} M_{{\cal
C}_{n_0, m_0}}^T)$ and obtains 
$\mbf{v}_{2,\san{T}_0} + \mbf{\hat{w}}_{2, \san{T}_0}$, 
where $n_0 := |\san{T}_0|$ is the cardinality of the set, $\san{T}_0$.

Our two-way IR protocol, which is based on 
Vollbrecht and Vestraete's idea of
two-way EDP \cite{vollbrecht:05}, is quite similar to
the previously described two-way processing except for one significant change.
As is usual in information theory, if we allow  negligible error probability,
Alice does not need  to send the parity sequence, $\mbf{u}_1$, to Bob
to identify parity sequence $\mbf{w}_1$.
More precisely, Bob can decode $\mbf{w}_1$ 
with negligible decoding error probability if Alice sends
a syndrome, $\mbf{t}_1 := \mbf{u}_1 M_{{\cal C}_{n,m}}^T$, for a linear code
such that the rate is $\frac{m}{n} \simeq H(P_{W_1})$ \cite[Corollary 2]{csiszar:82}.
Since Eve's available information from syndrome $\mbf{t}_1$ is
much smaller than that from sequence $\mbf{u}_1$ itself,
our IR protocol is more efficient than the above-mentioned two-way
processing in most cases, which will be discussed in 
Section \ref{annalysis-of-key-rate}.
Our IR protocol is formally executed as follows,
where the tilde $\tilde{}$ and hat $\hat{}$ on a sequence, a set or a
number 
indicate that
they are  guessed versions of those without these superscripts. 
Note that 
the inputs of the IR protocol are Alice's bit sequence $\mbf{x}$ and
Bob's bit sequence $\mbf{y}$, and the outputs of the IR protocol
are a sequence, $\hat{\mbf{u}}$, guessed by Alice and a sequence,
$\tilde{\mbf{u}}$, guessed by Bob.

\begin{enumerate}
\renewcommand{\theenumi}{\roman{enumi}}
\renewcommand{\labelenumi}{(\theenumi)}

\item \label{step1}
Alice locally computes $\mbf{u}_1$ and Bob does the same for $\mbf{v}_1$.

\item \label{step2}
For a prescribed $[n,n-m]$ linear code, ${\cal C}_{n,m}$, Alice
sends syndrome $\mbf{t}_1 = \mbf{u}_1 M_{{\cal C}_{n,m}}^T$
to Bob.

\item \label{step3}
Bob decodes $\hat{\mbf{w}}_1 := 
g_{{\cal C}_{n,m}}(\mbf{t}_1 + \mbf{v}_1 M_{{\cal C}_{n,m}}^T)$,
and sends $\hat{\mbf{w}}_1$ to Alice.

\item \label{step4}
Alice computes $\hat{\mbf{u}}_2$.
If the number, $\hat{n}_0 := |\{ i \mid \hat{w}_{i1} = 0 \}|$,
of blocks such that the guessed parity, $\hat{w}_{i1}$, of the discrepancies
is $0$ does not satisfy $\underline{n}_0 \le \hat{n}_0 \le \overline{n}_0$
for prescribed 
integers, $\underline{n}_0$ and $\overline{n}_0$,
then Bob randomly guesses $\hat{\mbf{u}}_{2, \hat{\san{T}}_0}$.
Otherwise, Alice sends the syndrome,
$\hat{\mbf{t}}_2 := \hat{\mbf{u}}_{2,\hat{\san{T}}_0} M_{{\cal C}_{\hat{n}_0,\hat{m}_0}}$,
for a prescribed $[\hat{n}_0, \hat{n}_0 - \hat{m}_0]$ linear code, 
${\cal C}_{\hat{n}_0, \hat{m}_0}$.

\item \label{step5}
Bob decodes $\tilde{\mbf{w}}_{2, \hat{\san{T}}_0} :=  
g_{{\cal C}_{\hat{n}_0, \hat{m}_0}}(\hat{\mbf{t}}_2 
+ \hat{\mbf{v}}_{2, \hat{\san{T}}_0} M_{{\cal C}_{\hat{n}_0, \hat{m}_0}}^T)$,
and obtains $\tilde{\mbf{u}}_{2, \hat{\san{T}}_0} :=  \hat{\mbf{v}}_{2, \hat{\san{T}}_0}
+ \tilde{\mbf{w}}_{2, \hat{\san{T}}_0}$.

\end{enumerate}

Note that $\hat{\mbf{u}}_{2, \hat{\san{T}}_1}$ and $\hat{\mbf{v}}_{2, \hat{\san{T}}_1}$
are set to all $0$s in our protocol, which is mathematically
equivalent to discarding them.

According to the universal channel coding theorem for the linear code
\cite[Corollary 2]{csiszar:82},
rates $\frac{m}{n} = H(P_{W_1}) + \delta$ and
$\frac{\hat{m}_0}{\hat{n}_0} = H(P_{W_2|W_1 = 0}) + \delta$ 
for  small $\delta > 0$ are sufficient for
Bob to decode $\mbf{w}_1$ and $\mbf{w}_{2, \san{T}_0}$ 
with negligible decoding error probability.
Furthermore, we set $\underline{n}_0 := n(P_{W_1}(0) - \delta)$
and $\overline{n}_0 := n ( P_{W_1}(1) + \delta)$ to 
satisfy the condition, $\underline{n}_0 \le \hat{n}_0 \le \overline{n}_0$,
in Step (\ref{step4}) with high probability.

\begin{remark}
Since we cannot estimate the probability distribution
of error exactly in QKD protocols and the actual distribution fluctuates around the
estimated error distribution, universality of codes is required.
Even though the distribution of errors in the QKD protocols are not necessarily
i.i.d., 
it is sufficient
to consider a universality condition on codes for the i.i.d. case.
More precisely, it is sufficient to use a linear code such that
the decoding error probability of the linear code is universally small 
for any binary symmetric channel 
whose crossover probability is close to the estimated error rate.
Such observations were first pointed out by Hamada \cite{hamada:04}.
Efficiently decodeable linear codes such as the low density parity check matrix
 code \cite{gallager:63} and the turbo code \cite{berrou:96}
satisfy this condition.
\end{remark}

\section{Security of QKD and key rate}
\label{security-of-qkd-and-key-rate}
 
This section presents the asymptotic key rate of QKD protocols
that employs the IR protocol proposed in Section 
\ref{two-way-information-reconciliation-protocol}.
The asymptotic key rate is derived by
the security proof method \cite{kraus:05, renner:05, renner:05b}.

We implement a prepare and measure scheme in a practical QKD protocol.
However, when we analyze the security of a QKD protocol, it is usually 
more convenient to consider its entanglement-based version.
Without compromising security, we can assume that Alice and Bob's raw keys
and bit sequences for error estimation are obtained by measuring a 
bipartite state, $\rho_{A^N B^N}$, on an $N$ pair of bipartite systems
$({\cal H}_A \otimes {\cal H}_B)^{\otimes N}$,
that $\rho_{A^N B^N}$ is invariant under the permutation of the 
systems \footnote{By applying the random permutation after the
transmission phase of QKD protocols, we can assume that Alice and
Bob's bit sequences are invariant under the permutation without
any compromise of the security.},
and that Eve can access  $\rom{Tr}_{A^N B^N} [\rho_{A^N B^N E^N}]$ for
a purification $\rho_{A^N B^N E^N}$ of $\rho_{A^N B^N}$ 
(see also \cite{kraus:05, renner:05}).
The specific form of $\rho_{A^N B^N}$ depends on which scheme Alice and Bob employ to
transmit a binary sequence, noise in the channel, and Eve's attack.
From \cite[Lemma 4.2.2]{renner:05b}, without loss of generality,
we can assume that  purification $\rho_{A^N B^N E^N}$ lies on the
symmetric subspace 
of $({\cal H}_A \otimes {\cal H}_B \otimes {\cal H}_E)^{\otimes N}$,
because any purification 
can be transformed into another purification 
using Eve's local operation.

Before the protocol is started, Alice and Bob discard the last $k$ subsystems, 
${\cal H}_A^{\otimes k} \otimes {\cal H}_B^{\otimes k}$, 
for  technical reasons of security proof.
More specifically, $k$ subsystems are discarded to apply
the de Finetti style representation theorem
\cite[Theorem 4.3.2]{renner:05b} (see also \cite{renner:07}) in the security proof.
Therefore, we set $N:= 2n + m + k$. Then, Alice and Bob conduct the protocol
for the state, $\rho_{A^{2n+m}B^{2n+m}} := \rom{Tr}_k[ \rho_{A^N B^N}]$, 
where $k$ is the number of discarded
systems, $m$ is the number of systems for 
parameter estimation,
and $2n$ is the number of systems that are used for key distillation.

First, Alice and Bob undertake the following parameter estimation protocol for the last
$m$-subsystems of the state $\rho_{A^{2n+m} B^{2n+m}}$.
The parameter estimation protocol is conducted to
estimate the number of discrepancies between Alice and
Bob's raw keys, and the amount of information that Eve has gained 
by eavesdropping.
\begin{enumerate}
\renewcommand{\theenumi}{\roman{enumi}}
\renewcommand{\labelenumi}{(\theenumi)}

\item
Alice and Bob carry out a bipartite 
positive operator valued measurement (POVM),
${\cal M} := \{ M_a \}_{a \in {\cal A}}$,
for each system, ${\cal H}_A \otimes {\cal H}_B$,
where ${\cal A}$ is the set of measurement outcomes.
The specific form of ${\cal M}$ depends on which scheme
we use. 

\item
If the type, 
$P_{\mbf{a}}$, of the measurement outcomes, 
$\mbf{a} =(a_1,\ldots, a_m)$, satisfies
$P_{\mbf{a}} \in {\cal Q}$
for a prescribed set, ${\cal Q}$,
the protocol outputs the type, $Q := P_{\mbf{a}}$, 
and Alice and Bob conduct the key distillation protocol
according to $Q$, where the type of sequence 
$\mbf{a} = (a_1,\ldots, a_m)$ is 
the frequency distribution defined by 
\begin{eqnarray*}
P_{\mbf{a}}(a) := \frac{ |\{i \mid 1 \le i \le m,~ a_i = a \}|}{m}
\mbox{ for }
a \in {\cal A}
\end{eqnarray*}
(for more details on the type, see \cite[Chapter 11]{cover}).
Otherwise, it outputs ``abort''.
\end{enumerate}
It is convenient to describe
the parameter estimation protocol using a 
completely positive (CP) map
as follows.
Let ${\cal M}^{\otimes m} := \{ M_{\mbf{a}} \}_{\mbf{a} \in {\cal A}^m}$ be 
a product POVM on $({\cal H}_A \otimes {\cal H}_B)^{\otimes m}$,
where $M_{\mbf{a}} = M_{a_1} \otimes \cdots \otimes M_{a_m}$.
Then, we can define a CP map, ${\cal E}_Q$, by
\begin{eqnarray}
{\cal E}_Q :
\rho_m 
\mapsto 
\sum_{\mbf{a} \in {\cal T}_Q^n({\cal A})} \rom{Tr} M_{\mbf{a}} \rho_m,  
    \label{eq-cpm-describing-parameter-estimation}
\end{eqnarray}
which maps the density operator to the probability such
that the parameter estimation protocol outputs $Q$,
where ${\cal T}_Q^m({\cal A})$ is a set of all sequences on ${\cal A}^m$
with type $Q$.

When the output of the parameter estimation 
protocol is $Q \in {\cal Q}$, 
Alice, Bob, and Eve's tripartite state is given by
\begin{eqnarray*}
\lefteqn{ \rho_{A^{2n} B^{2n} E^N}^Q
:=
\frac{1}{P_{\san{PE}}(Q)} }\\
&& (\rom{id}_{A^{2n} B^{2n}} \otimes 
{\cal E}_Q \otimes \rom{id}_{E^N})(\rho_{A^{2n+m}
B^{2n+m} E^N}),
\end{eqnarray*}
where $P_{\san{PE}}(Q)$ is a probability such that
the parameter estimation protocol outputs $Q$,
and $\rom{id}$ denotes the identity map on each system.
 
Alice and Bob apply a measurement 
${\cal M}_{XY} := \{M_x \otimes M_y \}_{(x,y) \in \mathbb{F}_2 \times \mathbb{F}_2}$
on ${\cal H}_A \otimes {\cal H}_B$ to
the remaining $2n$ systems
to obtain classical data (raw keys).
Then, Alice and Bob's measurement results, 
$(\mbf{x}, \mbf{y}) \in \mathbb{F}_2^{2n} \times \mathbb{F}_2^{2n}$, 
and Eve's available information 
is described by a $\{ccq\}$-state 
\footnote{A $\{ccq\}$-state is a tripartite state such that 
the first and second systems are classical and the third system
is quantum. See \cite{devetak:04} for a detail of this notation.}
\begin{eqnarray*}
\rho_{\mbf{X} \mbf{Y} E^N}^Q := ({\cal E}_{XY}^{\otimes 2n} \otimes
 \rom{id}_{E^N})(\rho_{A^{2n} B^{2n} E^N}^Q),
\end{eqnarray*}
where we introduce a CP map, ${\cal E}_{XY}$, that describes 
the measurement procedure for convenience.

According to  output $Q$ of the parameter estimation 
protocol, Alice and Bob decide the parameters of the 
IR protocol: rate $R(Q) := \frac{m}{n}$ of linear code 
${\cal C}_{n,m}$, numbers $\underline{n}_0(Q)$ and
$\overline{n}_0(Q)$ that are used in Step (\ref{step4}),
and rate $R_0(Q) := \frac{m_0}{n_0}$ of
linear code ${\cal C}_{n_0, m_0}$ for
$\underline{n}_0(Q) \le n_0 \le \overline{n}_0(Q)$.
Furthermore, Alice and Bob also decide the length, $\ell(Q)$,
of the finally distilled key according to $Q$.
According to the determined parameters, 
a final secure key pair is distilled as follows.
\begin{enumerate}
\renewcommand{\theenumi}{\roman{enumi}}
\renewcommand{\labelenumi}{(\theenumi)}

\item 
\label{key-distillation-step-1}
Alice and Bob undertake the two-way IR 
protocol in Section 
\ref{two-way-information-reconciliation-protocol},
and Alice obtains $\hat{\mbf{u}}$ and Bob obtains
$\tilde{\mbf{u}}$.

\item 
\label{key-distillation-step-2}

Alice and Bob carry out a 
privacy amplification (PA) protocol 
to distill a key pair $(s_A, s_B)$ such 
that Eve has little information about it.
Alice first randomly chooses a hash function,
$f: \mathbb{F}_2^{2n} \to \{ 0,1\}^{\ell(Q)}$,
from a family of two-universal hash functions 
(refer \cite[Definition 5.2.1]{renner:05b} for a formal
definition of a family of two-universal hash functions),
and sends the choice of $f$ to Bob over the public channel.
Then, Alice's distilled key is
$s_A = f(\hat{\mbf{u}})$ and 
Bos's distilled key is $s_B = f(\tilde{\mbf{u}})$.
\end{enumerate}

The distilled key pair and Eve's available information
can be described by a $\{cccq\}$-state,
$\rho_{S_A S_B C E^N}^Q$, where classical system
$C$ consists of random variables
$(\mbf{T}_1,\hat{\mbf{T}}_2,\hat{\mbf{W}}_1)$ that describe
 the exchanged messages 
$(\mbf{t}_1, \hat{\mbf{t}}_2, \hat{\mbf{w}}_1)$
in the IR protocol
and random variable $F$ that describes the choice  of the hash function in the PA protocol.
To define the security of the distilled key pair
$(S_A, S_B)$, we use the universally composable security 
definition \cite{ben-or:04, renner:05c},
which is defined by the trace distance between the
actual key pair and the ideal key pair.
We cannot state security in QKD protocols 
in the sense that the distilled key
pair $(S_A, S_B)$ is secure 
for a particular output $Q$ of the parameter estimation
protocol, because there is a slight possibility that
the parameter estimation protocol will not output ``abort''
even though Eve has so much information.
The QKD protocol is said to be $\varepsilon$-secure (in the sense of the average
over the outputs of the parameter estimation protocol) if
\begin{eqnarray}
      \label{eq-varepsilon-secure-average}
\sum_{Q \in {\cal Q}} P_{\san{PE}}(Q) 
\frac{1}{2}
\|
\rho^Q_{S_A S_B C E^N} - \rho_{S_A S_B}^{Q,\rom{mix}} 
\otimes \rho_{C E^N}
\| \le \varepsilon,
\end{eqnarray}
where $\rho_{S_A S_B}^{Q,\rom{mix}} := \sum_{s \in {\cal S}_Q}
\frac{1}{|{\cal S}_Q|} \ket{s,s}\bra{s,s}$
is the uniformly distributed key on the key space
${\cal S}_Q := \{0, 1\}^{\ell(Q)}$.

To state the relation between the security
and the asymptotic key rate
of the previously mentioned QKD protocol, define 
\begin{eqnarray*}
\Gamma(Q) := 
\{ \sigma_{AB} \mid
P_A^{\sigma_{AB}} = Q  \}
\end{eqnarray*}
as the set of two-qubit density operators that are
compatible with output $Q$ of the 
parameter estimation protocol,
where $P_A^{\sigma_{AB}}$ denotes the probability distribution
of the outcomes when measuring $\sigma_{AB}$ with POVM ${\cal M}$,
i.e., $P_A^{\sigma_{AB}}(a) := \rom{Tr} [M_a \sigma_{AB}]$.
For a purification, $\sigma_{ABE}$, of a density operator,
$\sigma_{AB} \in \Gamma(Q)$, let
$\sigma_{X_1 X_2 Y_1 Y_2 E_1 E_2} := ({\cal E}_{XY}^{\otimes 2}
\otimes \rom{id}_{E}^{\otimes 2})(\sigma_{ABE}^{\otimes 2})$ be 
a $\{ccq\}$-state that consists of $2$-bit pairs
$((X_1, X_2), (Y_1, Y_2))$ and environment systems $E_1, E_2$.
By using functions $\xi_1$ and $\xi_2$,
define random variables $(U_1, U_2, W_1, W_2)$
for the pair of bits $((X_1, X_2), (Y_1, Y_2))$
in the same way as in Section
\ref{two-way-information-reconciliation-protocol}.
Then, let $\sigma_{U_1 U_2 W_1 E_1 E_2}$ and 
$\sigma_{U_1 U_2 W_1 U_1 E_1 E_2}$ be
density operators that respectively describe the classical
random variables $(U_1, U_2, W_1)$
and $(U_1, U_2, W_1, U_1)$ with the environment
system $E_1, E_2$.

\begin{theorem}
     \label{theorem-assymptotic-key-rate}
For $Q \in {\cal Q}$, i.e., the output of the parameter
estimation protocol such that the QKD protocol does not abort,
let $\frac{\ell(Q)}{2n}$ be the key rate of the protocol.
For any $\varepsilon > 0$,
if the key rate satisfies 
\begin{eqnarray}
 \frac{\ell(Q)}{2n} &<& \frac{1}{2} \min_{\sigma_{AB} \in
 \Gamma(Q)}  
 \max\Bigl[
H_{\sigma}(U_1 U_2 | W_1 E_1 E_2) \nonumber \\
&&  - 
H(P_{W_1}) - P_{W_1}(0) H(P_{W_2|W_1=0}) , 
\label{eq-key-rate-formula} \\
&&  H_{\sigma}(U_2 | W_1 U_1 E_1
 E_2) - P_{W_1}(0) H(P_{W_2|W_1=0})
\Bigr], \nonumber 
\end{eqnarray}
then there exists a protocol that is $\varepsilon$-secure
in the sense of Eq.~(\ref{eq-varepsilon-secure-average})
for sufficiently large $n$, where 
$H_{\rho}(A|B) := H(\rho_{AB}) - H(\rho_B)$ is  
conditional von Neumann entropy \cite{hayashi-book:06},
and $H(P)$ is Shannon entropy \cite{cover}.
\end{theorem}

The meaning of the two arguments of the maximum
in Eq.~(\ref{eq-key-rate-formula}) should be noted.
The first argument 
states that the key rate is given by the difference between
Eve's ambiguity, $H_{\sigma}(U_1 U_2 | W_1 E_1 E_2)$,
about Alice's reconciled key and
the amount, $H(P_{W_1}) + P_{W_1}(0) H(P_{W_2|W_1=0})$,
of information leaked in the IR protocol.
On the other hand, since information leaked from
the syndrome, $\mbf{t}_1 = \mbf{u}_1 M_{{\cal C}_{n,m}}^T$, 
cannot be more than $\mbf{u}_1$ itself, 
we can evaluate the key rate 
under the condition that Eve can access $\mbf{u}_1$ itself, i.e.,
Eve's ambiguity, $H_{\sigma}(U_2 | W_1 U_1 E_1 E_2)$, about Alice's 
reconcilied key and the amount, $P_{W_1}(0) H(P_{W_2|W_1=0})$,
of information leaked in the IR protocol.
If either of them is omitted, 
the key rate is underestimated,
which will be discussed in Section \ref{annalysis-of-key-rate}.

Theorem \ref{theorem-assymptotic-key-rate} is formally proved by
demonstrating the above intuition formally, where we use 
a security proof method \cite{kraus:05,renner:05, renner:05b}.
More precisely, we use the techniques of privacy amplification 
and minimum entropy, and the de Finetti style representation theorem
and the property of symmetric states
(see \cite{renner:05b}).
Since the techniques used in the proof are not new and
involved, we give the proof for Theorem
\ref{theorem-assymptotic-key-rate} in 
the Appendix.

\section{Analysis of key rate}
\label{annalysis-of-key-rate}

Here, we analyze the asymptotic key rate formula
in Theorem \ref{theorem-assymptotic-key-rate}. 
More precisely, 
we derive a specific form of the key rate formulas as 
functions of the error rates
for the six-state  \cite{bruss:98} and  
BB84 protocols \cite{bennett:84}.

Before analyzing the key rate, let us define some notations.
For $\san{x},\san{z} \in \mathbb{F}_2$, let 
\begin{eqnarray*}
\ket{\psi(\san{x},\san{z})} :=
\frac{1}{\sqrt{2}}(
\ket{0}\ket{0 + \san{x}} + (-1)^{\san{z}}
\ket{1} \ket{1 + \san{x}}
)
\end{eqnarray*}
be the Bell states on two-qubit systems ${\cal H}_A \otimes {\cal H}_B$.
For a probability distribution, 
$P_{\san{XZ}}$, on $\mathbb{F}_2 \times \mathbb{F}_2$,
a state of the form, 
\begin{eqnarray*}
\sum_{\san{x},\san{z} \in \mathbb{F}_2 } P_{\san{XZ}}(\san{x},\san{z})
\ket{\psi(\san{x},\san{z})}\bra{\psi(\san{x},\san{z})},
\end{eqnarray*}
is called a Bell diagonal state.
We occasionally abbreviate $P_{\san{XZ}}(\san{x}, \san{z})$
as $p_{\san{xz}}$.

\begin{theorem}
    \label{theorem-key-rate-vs-error-rate}
For a Bell diagonal state, 
$\sigma_{AB} = \sum_{\san{x}, \san{z} \in \mathbb{F}_2} P_{\san{XZ}}(\san{x},\san{z})
\ket{\psi(\san{x},\san{z})}\bra{\psi(\san{x},\san{z})}$, we have
\begin{eqnarray}
\lefteqn{ \frac{1}{2} 
\max[
H_{\sigma}(U_1 U_2 | W_1 E_1 E_2) } \nonumber \\
&& - 
H(P_{W_1}) - P_{W_1}(0) H(P_{W_2|W_1=0}) , \nonumber \\
&& H_{\sigma}(U_2 | W_1 U_1 E_1
 E_2) - P_{W_1}(0) H(P_{W_2|W_1=0})
] \nonumber \\
&=& \max[
1 - H(P_{\san{XZ}}) \nonumber \\
&&  + \frac{P_{\bar{\san{X}}}(1)}{2}
h\left( \frac{p_{00} p_{10} + p_{01} p_{11}}{
(p_{00} + p_{01})(p_{10} + p_{11})}\right), \nonumber \\
&& \hspace{10mm}
\frac{P_{\bar{\san{X}}}(0)}{2} (
1 - H(P_{\san{XZ}}^\prime) ) ] ,
\label{eq-key-rate-vs-error-rate}
\end{eqnarray} 
where $h(p) := - p\log p - (1-p) \log (1-p)$ is the 
binary entropy function,
\begin{eqnarray*}
P_{\bar{\san{X}}}(0) &:=& (p_{00} + p_{01})^2
+ (p_{10} + p_{11})^2, \\
P_{\bar{\san{X}}}(1) &:=& 2 (p_{00} + p_{01})
(p_{10} + p_{11}),
\end{eqnarray*}
and
\begin{eqnarray*}
P_{\san{XZ}}^\prime(0,0) &:=&
\frac{p_{00}^2 + p_{01}^2}{(p_{00} + p_{01})^2 + (p_{10}+p_{11})^2}, \\
P_{\san{XZ}}^\prime(1,0) &:=& 
\frac{2 p_{00} p_{01}}{(p_{00} + p_{01})^2 + (p_{10}+p_{11})^2}, \\
P_{\san{XZ}}^\prime(0,1) &:=&
\frac{p_{10}^2 + p_{11}^2}{(p_{00} + p_{01})^2 + (p_{10}+p_{11})^2}, \\
P_{\san{XZ}}^\prime(1,1) &:=&
\frac{2 p_{10}p_{11} }{(p_{00} + p_{01})^2 + (p_{10}+p_{11})^2}.
\end{eqnarray*}
\end{theorem}
The theorem is proved by a straight forward calculation.
Thus, the proof is presented in the Appendix E.

The six-state protocol \cite{bruss:98} uses three different bases defined by 
$z$-basis $\{ \ket{0_z}, \ket{1_z} \}$, $x$-basis 
$\{ 1/\sqrt{2}( \ket{0_z} \pm \ket{1_z}) \}$,
and $y$-basis $\{ 1/\sqrt{2} ( \ket{0_z} \pm i \ket{1_z}) \}$.
When Alice and Bob obtain an error rate, $e$, 
the set $\Gamma(Q)$ consists of states whose Bell diagonal entries 
$p_{00},p_{10},p_{01}, p_{11}$ satisfy conditions
$p_{10} + p_{11} = e$, $p_{01} + p_{11} = e$, and $p_{01} + p_{10} = e$.
Together with the normalization condition, we find
$p_{00} = 1 - \frac{3e}{2}$ and $p_{10} = p_{01} = p_{11} =
\frac{e}{2}$.
Since it is sufficient only to minimize over the Bell diagonal states
(see the Appendix F),
the key rate of the six-state protocol for the error rate $e$ is given
by substituting $p_{00} = 1 - \frac{3e}{2}$ and $p_{10} = p_{01} = p_{11} =
\frac{e}{2}$ into Eq.~(\ref{eq-key-rate-vs-error-rate}).
The key rate of the six-state protocol that uses the proposed IR protocol 
is plotted in Fig. \ref{Fig-six-state}.
\begin{figure}
\centering
\includegraphics[width=\linewidth]{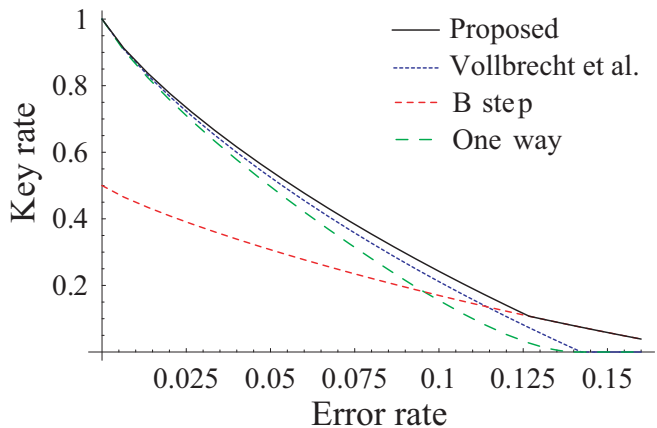}
\caption{ (Color online)
Comparison of the key rates of the  
six-state protocols. 
``Proposed'' is
the key rate of the six-state protocol that uses
the proposed IR protocol. 
``Vollbrecht et al.'' is the key rate of the two-way six-state protocol
of \cite{ma:06, watanabe:06}.
``B-step'' is the key rate of the two-way six-state protocol of
 \cite{gottesman:03}.
``One-way'' is the key rate of
the one-way six-state protocol with the noisy preprocessing \cite{renner:05}.
It should be noted that the key rates of two-way six-state protocols
of \cite{renner:05b, gottesman:03, chau:02} are slightly higher
than that of the proposed protocol for much higher error rate.
}
\label{Fig-six-state}
\end{figure}

The BB84 protocol is similar to the six-state protocol, but only uses 
the $z$-basis and the $x$-basis to transmit a bit sequence.
Thus, we only obtain two conditions on the four coefficients $p_{00}, p_{10}, p_{01}, p_{11}$.
Thus, the set $\Gamma(Q)$ consists of states whose Bell diagonal entries
satisfy conditions
$p_{10} + p_{11} = e$ and $p_{01} + p_{11} = e$.
The resulting candidates for Bell diagonal states 
in $\Gamma(Q)$ have coefficients
$p_{00} = 1 - 2e + p_{11}$, $p_{10} = p_{01} = e - p_{11}$, and $p_{11} \in [0, e]$,
and we have to minimize 
the key rate formula of Eq.~(\ref{eq-key-rate-vs-error-rate})
over the free parameter, $p_{11} \in [0,e]$.
The key rate of the BB84 protocol that uses the proposed IR protocol 
is plotted in Fig. \ref{Fig-bb84}.
\begin{figure}
\centering
\includegraphics[width=\linewidth]{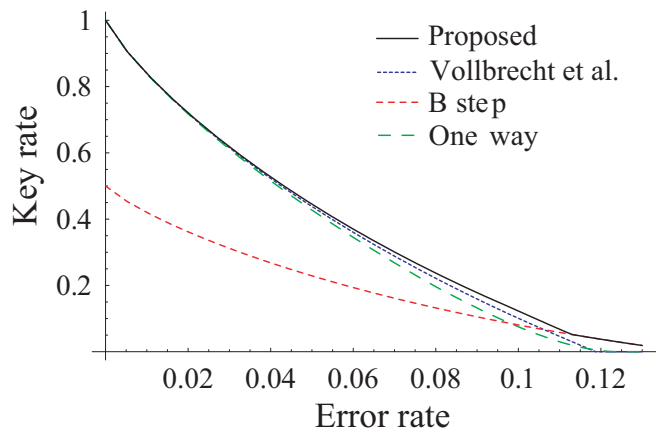}
\caption{ (Color online)
Comparison of the key rates of the  
BB84 protocols. 
``Proposed'' is
the key rate of the BB84 protocol that uses
the proposed IR protocol. 
``Vollbrecht et al.'' is the key rate of the two-way BB84 protocol
of \cite{ma:06, watanabe:06}.
``B-step'' is the key rate of the two-way BB84 protocol of
 \cite{gottesman:03}.
``One-way'' is the key rate of
the one-way BB84 protocol with the noisy preprocessing \cite{renner:05}.
}
\label{Fig-bb84}
\end{figure}
\begin{remark}
     \label{remark-relation-to-b-step}
By using the chain rule of von Neumann entropy, 
we can rewrite the l.h.s. of 
Eq.~(\ref{eq-key-rate-vs-error-rate}) as
\begin{eqnarray}
\lefteqn{ \frac{1}{2}\{
 \max[ H_{\sigma}(U_1|W_1 E_1 E_2) - H(P_{W_1}), 0] } \nonumber \\
&+& \hspace{-3mm} H_{\sigma}(U_2| W_1 U_1 E_1 E_2) -
P_{W_1}(0) H(P_{W_2|W_1=0})
\}.
\label{eq-key-rate-rewritten}
\end{eqnarray}
We can interpret this formula as follows.
If Bob's ambiguity, $H(P_{W_1})$, about  bit $U_1$,
i.e., the amount of transmitted syndrome per bit, is smaller
than Eve's ambiguity, $H_{\sigma}(U_1|W_1 E_1 E_2)$, about
bit $U_1$, then Eve cannot decode sequence $\mbf{U}_1$
\cite{slepian:73, devetak:03}, and there exists some remaining ambiguity
about bit $U_1$ for Eve.
We can thus distill some secure key from bit $U_1$.
On the other hand, if Bob's ambiguity, $H(P_{W_1})$,
about bit $U_1$,
i.e., the amount of transmitted syndrome per bit, is larger
than Eve's ambiguity, $H_{\sigma}(U_1|W_1 E_1 E_2)$,
about $U_1$, then Eve might be able
to decode sequence $\mbf{U}_1$ from
her side information, $W_1$, $E_1$, $E_2$, and 
the transmitted syndrome \cite{slepian:73, devetak:03}.
Thus, there exists  the possibility that Eve can completely know
bit $U_1$, and we can distill no secure key from bit $U_1$,
because we have to consider the worst case in a cryptographic 
scenario.
Consequently, sending the hashed version (syndrome) of sequence
 $\mbf{U}_1$ instead of $\mbf{U}_1$ itself is
not always effective, and the slopes of the key rate curves 
in Figs.~\ref{Fig-six-state} and \ref{Fig-bb84} change
when Eve becomes able to decode $\mbf{U}_1$.

The second and third terms of
Eq.~(\ref{eq-key-rate-rewritten})  are the same
as the key rate formula of the protocol that
uses Gottesman and Lo's  B-step \cite{gottesman:03} followed
by error correction and privacy amplification.
Even though Alice sends the sequence $\mbf{U}_1$ itself
instead of its hashed version in the B-step,
the key rate of the protocol with the B-step
is equal to that
of the proposed protocol for high error rates, because
Eve can decode sequence $\mbf{U}_1$ from her side
information and the transmitted syndrome.
\end{remark}
\begin{remark}
     \label{remark-relation-to-vollbrecht}
The yield of Vollbrecht and Vestraete's EDP \cite{vollbrecht:05} and the key 
rate of the QKD protocols \cite{ma:06, watanabe:06} are given by
\begin{eqnarray}
\lefteqn{ 1 - H(P_{\san{XZ}}) } \nonumber \\
&+& \frac{P_{\bar{\san{X}}}(1)}{4} \left\{
 h\left( 
\frac{p_{01}}{p_{00} + p_{01} }
\right) 
+ h\left(
\frac{p_{11}}{p_{10} + p_{11}}
\right) \right\}.
      \label{eq-rate-of-volbrecht}
\end{eqnarray}
We can find by the concavity of the binary entropy function
that the first argument in the maximum of the r.h.s.
of Eq.~(\ref{eq-key-rate-vs-error-rate}) is larger than
the value in Eq.~(\ref{eq-rate-of-volbrecht}).
To explain why the key rate of the proposed protocol is 
higher than that of \cite{ma:06, watanabe:06},
we need to review the EDP \cite{vollbrecht:05} by
using the notations in Section
 \ref{two-way-information-reconciliation-protocol}.
Assume that Alice and Bob share Bell diagonal states,
$\sigma_{AB}^{\otimes 2n}$. First, Alice and Bob
divide $2n$ pairs into $n$ blocks of length $2$, and locally carry out
CNOT operation  on each block,
where the $2i$-th pair is the source and $(2i-1)$-th
pair is the target.
Then, Alice and Bob undertake the breeding protocol \cite{bennett:96} 
to guess bit flip errors in the $(2i-1)$-th pair for all $i$.
The guessed bit flip errors can be described by a sequence,
$\hat{\mbf{w}}_1$. Note that two-way communication is used in this step.
According to sequence $\hat{\mbf{w}}_1$, Alice and Bob classify
 indices of blocks into two sets, $\hat{\san{T}}_0$ and
 $\hat{\san{T}}_1$.
For a collection of $2i$-th pairs such that $i \in \hat{\san{T}}_0$,
Alice and Bob conduct the breeding protocol 
to correct bit flip errors.
For a collection of $2i$-th pairs such that $i \in \hat{\san{T}}_1$,
Alice and Bob perform measurements by $\{ \ket{0_z}, \ket{1_z} \}$
basis, and obtain measurement results,
$\mbf{x}_{2, \hat{\san{T}}_1}$ and $\mbf{y}_{2, \hat{\san{T}}_1}$.
Alice sends $\mbf{x}_{2, \hat{\san{T}}_1}$ to Bob.
Alice and Bob correct the phase errors
for the remaining pairs 
by using information $\hat{\san{T}}_0$ and $\hat{\san{T}}_1$,
and bit flip error 
$\mbf{x}_{2, \hat{\san{T}}_1} + \mbf{y}_{2, \hat{\san{T}}_1}$.

If we convert this EDP into a QKD protocol, 
the difference between that QKD protocol and ours is as follows.
In the protocol converted from \cite{vollbrecht:05}, after Step (\ref{step3}), Alice
reveals the sequence, $\mbf{x}_{2, \hat{\san{T}}_1}$,
which consists of the second bit, $x_{i2}$, of the $i$-th
block such that the parity of discrepancies $\hat{w}_{i1}$
is $1$. 
However, Alice discards $\mbf{x}_{2, \hat{\san{T}}_1}$
in the proposed IR protocol of Section
 \ref{two-way-information-reconciliation-protocol}.
Since sequence $\mbf{x}_{2,\hat{\san{T}}_1}$ has 
some correlation to sequence $\mbf{u}_1$ from
the view point of Eve, Alice should not reveal 
$\mbf{x}_{2, \hat{\san{T}}_1}$ 
to achieve a higher key rate.

In the EDP context, on the other hand, since the bit flip error,
$\mbf{x}_{2, \hat{\san{T}}_1} + \mbf{y}_{2, \hat{\san{T}}_1}$, has some
 correlation to the phase flip errors in the $(2i-1)$-th pair with 
$i \in \hat{\san{T}}_1$, Alice should send the measurement results,
$\mbf{x}_{2, \hat{\san{T}}_1}$, to Bob. If Alice discards measurement
results $\mbf{x}_{2, \hat{\san{T}}_1}$ without telling Bob what
the result is, then the yield of the resulting EDP is worse than
Eq.~(\ref{eq-rate-of-volbrecht}). 
Consequently, there seems to be no correspondence between the EDP
and our proposed classical processing.
\end{remark}

\section{Conclusion}

We proposed an information reconciliation protocol that uses
two-way classical communication.
For the BB84 and six-state protocols,
the key rates of QKD protocols that uses our information
reconciliation protocol are higher than previously known
protocols for a wide range of error rates.
Furthermore, we showed the relation between the proposed
protocol and the B-step of \cite{gottesman:03}
(Remark \ref{remark-relation-to-b-step}).
We clarified why 
the key rate of our protocol is higher than
those of \cite{vollbrecht:05, ma:06, watanabe:06}
(Remark \ref{remark-relation-to-vollbrecht}),
and found that there does not seem to be any EDP that corresponds to
our proposed QKD protocol.

\section*{Acknowledgment}

The first author partly contributed to this work during
his internship at Nippon Telegraph Communication
Science Laboratories. 
This research was also partly
supported by the Japan Society for the Promotion
of Science under a Grants-in-Aid for Young
Scientists, No.~18760266,  and a Grants-in-Aid for
JSPS Fellows.

\appendix
\section{Notations}
\label{notation-and-review-of-known-results}

These appendices are suplementary materials, in which
we prove Theorem 2, Theorem 3, and the fact that the key rate formula
evaluated for a Bell-diagonal state is the worst case.
The proof of Theorem 2 
is based on the proof method of
\cite{kraus:05,renner:05, renner:05b},
especially \cite{renner:05b}.
In Section \ref{notation-and-review-of-known-results},
we review notations and fundamental results that
are used in subsequent sections.
Notations in this paper is almost the same as those in \cite{renner:05b}.
In Section \ref{sec-privacy-amplification},
we review notions of the (smooth) min-entropy,
the (smooth) max-entropy, and the privacy amplification.
Furthermore, we additionally show some lemmas, which are
used to prove Theorem 2 in Section \ref{sec-proof-of-theorem}.
In Section \ref{sec-symmetric-states},
we review the property of symmetric states
and the de Finetti style representation theorem \cite{renner:05b,
renner:07}.
We prove Theorem 2
in Section \ref{sec-proof-of-theorem}.
Section \ref{proof-of-theorem-3} presents a proof of
Theorem 3.
We show the fact that 
the key rate formula
evaluated for a Bell-diagonal state is the worst case
in Section \ref{bell-worst}.

\subsection{Fundamentals}

For a finite set ${\cal X}$, let ${\cal P}({\cal X})$ be 
the set of non-negative functions $P$ on ${\cal X}$, i.e.,
$P(x) \ge 0$ for all $x \in {\cal X}$.
If $P \in {\cal P}({\cal X})$ is normalized, 
i.e., $\sum_{x \in {\cal X}} P(x) = 1$,
then $P$ is a probability distribution on ${\cal X}$.
Unless stated as a probability distribution, 
$P \in {\cal P}({\cal X})$ is not necessarily normalized.

For a finite-dimensional Hilbert space ${\cal H}$, let 
${\cal P}({\cal H})$ be the set of non-negative operator
$\rho$ on ${\cal H}$.
If $\rho \in {\cal P}({\cal H})$ is normalized, i.e., $\rom{Tr} \rho = 1$,
then $\rho$ is called  a  density operator. 
Mathematically, a state of a quantum mechanical system with
$d$-degree of freedom is represented by a density operator on ${\cal H}$
with $\dim {\cal H} = d$.
Unless stated as a density operator or a state, 
$\rho \in {\cal P}({\cal H})$ is not necessarily normalized.
For Hilbert spaces ${\cal H}_A$ and ${\cal H}_B$, the set of non-negative
operators ${\cal P}({\cal H}_A \otimes {\cal H}_B)$ on the tensor product space
${\cal H}_A \otimes {\cal H}_B$ is defined in a similar manner.

The classical random variables can be regarded as a special case of
the quantum states. For a random variable $X$ with a distribution
$P_X \in {\cal P}({\cal X})$, let 
\begin{eqnarray*}
\rho_X := \sum_{x \in {\cal X}} P_X(x) \ket{x}\bra{x},
\end{eqnarray*}
where $\{ \ket{x} \}_{x \in {\cal X}}$ is an orthonormal basis of ${\cal H}_X$.
We call $\rho_X$ the operator representation of the classical
distribution $P_X$.

When a quantum system ${\cal H}_A$ is prepared in a state $\rho^x_A$ according
to a realization $x$ of a random variable $X$ with a probability distribution $P_X$, it is
convenient to denote 
it by a density operator
\begin{eqnarray}
\label{eq-definition-cq-state}
\rho_{XA} := \sum_{x \in {\cal X}} P_X(x) \ket{x}\bra{x} \otimes
 \rho_A^x \in {\cal P}({\cal H}_X \otimes {\cal H}_A),
\end{eqnarray}
where $\{ \ket{x} \}_{x \in {\cal X}}$ is an orthonormal basis of 
${\cal H}_X$.
We call the density operator $\rho_{XA}$ a $\{cq\}$-state \cite{devetak:04},
or we say $\rho_{XA}$ is classical on ${\cal H}_X$.
We call $\rho_A^x$ a conditional operator. 
When a quantum system ${\cal H}_A$ is prepared in a state $\rho_A^{x,y}$
according to a joint random variable $(X, Y)$ with a probability
distribution $P_{XY}$, a state $\rho_{XYA}$ is defined in a similar
manner, and the state $\rho_{XYA}$ is called a $\{ ccq \}$-state. 

In quantum mechanics, the most general state evolution of a
quantum mechanical system is described by a
completely positive (CP) map. It can be shown that
any CP map ${\cal E}$ can be written as
\begin{eqnarray}
	\label{eq-kraus-operator}
{\cal E}(\rho) = \sum_{a \in {\cal A}} E_a \rho E_a^*
\end{eqnarray}
for a family of linear operators
$\{ E_a \}_{a \in {\cal A}}$ from
the initial system ${\cal H}$ to the destination system 
${\cal H}^\prime$.
We usually require the map to be trace preserving (TP),
i.e., $\sum_{a \in {\cal A}} E_a^* E_a = \rom{id}_{\cal H}$,
but if a state evolution involves a measurement,
then the corresponding CP map is not necessarily trace preserving, i.e.,
$\sum_{a \in {\cal A}} E_a^* E_a \le \rom{id}_{\cal H}$.

\subsection{Distance and fidelity}

In this paper, we use two kind of distances.
One is the variational distance of ${\cal P}({\cal X})$.
For non-negative functions 
$P, P^\prime \in {\cal P}({\cal X})$, the variational distance between
$P$ and $P^\prime$ is defined by
\begin{eqnarray*}
\| P - P^\prime \| := \sum_{x \in {\cal X}} | P(x) - P^\prime(x) |.
\end{eqnarray*}
The other distance used in this paper is 
the trace distance of ${\cal P}({\cal H})$.
For nen-negative operators $\rho, \sigma \in {\cal P}({\cal H})$, the trace distance
between $\rho$ and $\sigma$ is defined by
\begin{eqnarray*}
\| \rho - \sigma \| := \rom{Tr} | \rho - \sigma |,
\end{eqnarray*}
where $|A| := \sqrt{A^* A}$ for a operator on ${\cal H}$, and 
$A^*$ is the adjoint operator of $A$.
The following lemma states that the trace distance 
between (not necessarily normalized operators) does not 
increase by applying a CP map, and it is used several times 
in this paper.
\begin{lemma}
\cite[Lemma A.2.1]{renner:05b}
Let $\rho, \rho^\prime \in {\cal P}({\cal H})$ and let 
${\cal E}$ be a trace-non-increasing CP map, i.e., ${\cal E}$ satisfies
$\rom{Tr}{\cal E}(\sigma) \le \rom{Tr} \sigma$ for any
$\sigma \in {\cal P}({\cal H})$. Then we have
\begin{eqnarray*}
\| {\cal E}(\rho) - {\cal E}(\rho^\prime) \| \le \| \rho - \rho^\prime \|.
\end{eqnarray*}
\end{lemma}

The following lemma states that, for a $\{cq\}$-state 
$\rho_{XB}$, if two classical messages $v$ and $\bar{v}$ are
computed from $x$ and they are equal with high probability,
then the $\{ccq\}$ state $\rho_{XVB}$ and $\rho_{X\bar{V}B}$ that
involve computed classical messages $v$ and $\bar{v}$ 
are close with respect to the trace distance.
\begin{lemma}
	\label{lemma-error-prob-continuity}
Let 
\begin{eqnarray*}
\rho_{XB} := \sum_{x \in {\cal X}} P_X(x) \ket{x}\bra{x} \otimes \rho_B^x
\end{eqnarray*}
be a $\{cq\}$-state, and let $V := f(X)$ for a function $f$ and
$\bar{V} := g(X)$ for a function $g$. Assume that 
\begin{eqnarray*}
\Pr\{ V \neq \overline{V} \} =
\sum_{\scriptstyle x \in {\cal X}  \atop
f(x) \neq g(x)} P_X(x) 
\le \varepsilon.
\end{eqnarray*}
Then, for $\{ ccq\}$-states
\begin{eqnarray*}
\rho_{XVB} := \sum_{x \in {\cal X} }
P_X(x) \ket{x}\bra{x}  \otimes \ket{f(x)}\bra{f(x)}
\otimes \rho_B^x
\end{eqnarray*}
and 
\begin{eqnarray*}
\rho_{X\overline{V}B} := \sum_{x \in {\cal X} }
P_X(x) \ket{x}\bra{x}  \otimes \ket{g(x)}\bra{g(x)}
\otimes \rho_B^x,
\end{eqnarray*}
we have
\begin{eqnarray*}
\| \rho_{XVB} - \rho_{X\overline{V}B} \| \le 2 \varepsilon.
\end{eqnarray*}
\end{lemma}
\begin{proof}
We have
\begin{eqnarray*}
&& \| \rho_{XVB} - \rho_{X\overline{V}B} \| \\
&=& \sum_{x \in {\cal X} } P_X(x)
\| \ket{x}\bra{x} \| \\
&& ~~~\cdot \| \ket{f(x)}\bra{f(x)} - \ket{g(x)}\bra{g(x)} \|
\cdot \| \rho_B^x \| \\
&=& \sum_{x \in {\cal X} } P_X(x)
\cdot 2(1- \delta_{f(x), g(x)}) \\
&\le& 2 \varepsilon,
\end{eqnarray*}
where $\delta_{a,b} = 1$ if $a=b$ and $\delta_{a,b} = 0$ if $a \neq b$. 
\end{proof}

The fidelity between two (not necessarily normalized) operators 
$\rho, \sigma \in {\cal P}({\cal H})$
is defined by
\begin{eqnarray*}
F(\rho, \sigma) := \rom{Tr} \sqrt{\sqrt{\rho}\sigma \sqrt{\rho}}.
\end{eqnarray*}
The following lemma is an extension of Uhlmann's theorem to 
non-normalized operators $\rho$ and $\sigma$.
\begin{lemma}
\cite[Theorem A.1.2]{renner:05b}
	\label{uhlman-theorem}
Let $\rho, \sigma \in {\cal P}({\cal H})$,
and let $\ket{\psi} \in {\cal H}_R \otimes {\cal H}$ be a purification of $\rho$. Then
\begin{eqnarray*}
F(\rho, \sigma) = \max_{\ket{\phi}\bra{\phi}} F(\ket{\psi}\bra{\psi}, \ket{\phi}\bra{\phi}),
\end{eqnarray*}
where the maximum is taken over all purifications 
$\ket{\phi} \in {\cal H}_R \otimes {\cal H}$ of $\sigma$.
\end{lemma}
The trace distance and the fidelity have close relationship.
If the trace distance
between two density operators $\rho$ and $\sigma$ is close to $0$,
then the fidelity between $\rho$ and $\sigma$ is close to $1$,
and vise versa.
\begin{lemma}
\cite[Lemma A.2.4]{renner:05b}
	\label{upper-bound-by-fidelity}
Let $\rho, \sigma \in {\cal P}({\cal H})$. Then, we have
\begin{eqnarray*}
\| \rho - \sigma \| \le 
\sqrt{(\rom{Tr}\rho + \rom{Tr} \sigma)^2 - 4 F(\rho, \sigma)^2 }.
\end{eqnarray*}
\end{lemma}
\begin{lemma}
\cite[Lemma A.2.6]{renner:05b}
	\label{lower-bound-by-fidelity}
Let $\rho, \sigma \in {\cal P}({\cal H})$. Then, we have
\begin{eqnarray*}
\rom{Tr}\rho + \rom{Tr} \sigma - 2 F(\rho, \sigma) 
\le \| \rho - \sigma \|.
\end{eqnarray*}
\end{lemma}

\subsection{Entropy}

For a random variable $X$ on ${\cal X}$ with
a probability distribution $P_X \in {\cal P}({\cal X})$, the entropy
of $X$ is defined by
\begin{eqnarray*}
H(X) = H(P_X) := - \sum_{x \in {\cal X}} P_X(x) \log P_X(x),
\end{eqnarray*}
where the base of $\log$ is $2$.
Especially for a real number $0 \le p \le 1$,
the binary entropy function is defined by 
\begin{eqnarray*}
h(p) := - p \log p - (1-p) \log (1-p).
\end{eqnarray*}

Similarly, for a joint random variables $X$ and $Y$ with
a joint  probability distribution $P_{XY} \in {\cal P}({\cal X} \times {\cal Y})$,
the joint entropy of $X$ and $Y$ is
\begin{eqnarray*}
H(XY) &=& H(P_{XY}) \\
&:=& - \sum_{(x,y) \in {\cal X} \times {\cal Y}}
P_{XY}(x,y) \log P_{XY}(x,y).
\end{eqnarray*}

The conditional entropy of $X$ given $Y$ is defined by
\begin{eqnarray*}
H(X|Y) := H(XY) - H(Y).
\end{eqnarray*}
For a quantum state $\rho \in {\cal P}({\cal H})$, the von Neumann entropy
of the system is defined by
\begin{eqnarray*}
H(\rho) := \rom{Tr} \rho \log \rho.
\end{eqnarray*}

For a quantum state $\rho_{AB} \in {\cal P}({\cal H}_A \otimes {\cal H}_B)$
of the composite system, the von Neumann entropy of the composite system
is $H(\rho_{AB})$.
The conditional von Neaumann entropy of the system $A$ given the system $B$
is defined by
\begin{eqnarray*}
H_{\rho}(A|B) := H(\rho_{AB}) - H(\rho_B),
\end{eqnarray*}
where $\rho_B = \rom{Tr}_A [ \rho_{AB} ]$ is the partial trace of
$\rho_{AB}$ over the system $A$.

\begin{remark}
    \label{remark-convention-2}
In this paper, we denote $\rho_A$ for  $\rom{Tr}_{B}[ \rho_{AB}]$ or
$\rho_{B}$ for $\rom{Tr}_{AC}[ \rho_{ABC}]$ e.t.c. without
declaring them if they are obvious from the context. 
\end{remark}

\subsection{Method of type}

In this section, we review the method of type
that are used in this paper
(see \cite[Chapter 11]{cover} for more detail).

For a sequence $\mbf{x} = (x_1, \ldots, x_n) \in {\cal X}^n$,
the type of $\mbf{x}$ is the 
empirical probability distribution $P_{\mbf{x}} \in {\cal P}({\cal X})$ defined by
\begin{eqnarray*}
P_{\mbf{x}}(a) := \frac{ | \{ i \mid x_i = a \} | }{n}~~~~~~\mbox{for }
a \in {\cal X},
\end{eqnarray*}
where $|A|$ is the cardinality of a set $A$.
Let 
\begin{eqnarray*}
{\cal P}_n({\cal X}) := \{ P_{\mbf{x}} \mid  \mbf{x} \in {\cal X}^n \}
\end{eqnarray*}
be the set of all types on ${\cal X}^n$.
It is easy to confirm that 
\begin{eqnarray*}
|{\cal P}_n({\cal X})| \le (n+1)^{(|{\cal X}| - 1)}.
\end{eqnarray*}
For $Q \in {\cal P}_n({\cal X})$, 
\begin{eqnarray*}
{\cal T}_Q^n({\cal X}) := \{ \mbf{x} \in {\cal X}^n \mid
P_{\mbf{x}} = Q \}
\end{eqnarray*}
is the set of all sequences of type $Q$.

The probability that sequences in the set ${\cal T}_Q^n$ occur can be
expressed in terms of the divergence.
\begin{lemma}
\cite[Theorem 11.1.4]{cover}
For any probability distribution $P \in {\cal P}({\cal X})$ and 
for any type $Q \in {\cal P}_n({\cal X})$, we have
\begin{eqnarray*}
\frac{1}{(n+1)^{(|{\cal X}|-1)}} \exp\{ -n D(Q\| P)\} 
&\le&  P^n({\cal T}_Q^n) \\
&\le&  \exp\{ -n D(Q \| P) \},
\end{eqnarray*}
where $P^n({\cal T}_Q^n) := \sum_{\mbf{x} \in {\cal T}_Q^n}
 P^n(\mbf{x})$, 
the base of $\exp\{\}$ is $2$,
and
$D(Q \| P)$ is the divergence defined by 
\begin{eqnarray*}
D(Q \| P) := \sum_{x \in {\cal X}} Q(x) \log \frac{Q(x)}{P(x)}.
\end{eqnarray*}
\end{lemma}

In the subsequent sections, we especially use the following inequality:
\begin{eqnarray}
	\label{eq-inequality-of-type}
Q^n({\cal T}_Q^n({\cal X})) \ge \frac{1}{(n+1)^{(|{\cal X}|-1)}}
\end{eqnarray}
for any $Q \in {\cal P}_n({\cal X})$,
which follows from the fact that $D(Q \| Q) = 0$.
\begin{lemma}
    \cite[Lemma 11.6.1]{cover}
For any probability distributions $P, P^\prime \in {\cal P}({\cal X})$,
 we have 
\begin{eqnarray*}
D(P \| P^\prime ) \ge \frac{1}{2 \ln 2} \| P - P^\prime \|^2.
\end{eqnarray*}
\end{lemma}
The following corollary states that sequences whose types are not close
to $P$ rarely occur as $n$ increases.
\begin{corollary}
 \label{corollary-type}
For any probability distribution $P \in {\cal P}({\cal X})$ and a
set ${\cal B}^{\varepsilon}(P) := \{ \mbf{x} \in {\cal X}^n \mid \| P_{\mbf{x}}
 - P \| \le \varepsilon \}$, we have
\begin{eqnarray*}
\sum_{\mbf{x} \notin {\cal B}^\varepsilon(P)} P^n(\mbf{x}) \le (n+1)^{(|{\cal
 X}| -  1)} \exp \left\{ - \frac{\varepsilon^2 n}{2 \ln 2} \right\}.
\end{eqnarray*}
\end{corollary}

\section{Privacy Amplification}
\label{sec-privacy-amplification}

In this section, we review the privacy amplification.
First, we review notions of the (smooth) min-entropy and 
the (smooth) max-entropy. The (smooth) min-entropy and the (smooth) max-entropy are
useful tool to prove the security of QKD protocol 
\cite{kraus:05, renner:05, renner:05b}.
Especially, (smooth) min-entropy is much more important,
because it is related to the length of the securely distillable
key by the privacy amplification.
The privacy amplification \cite{bennett:95} is a technique
to  distill a secret key from partially secret data, on
which an adversary might have some information.
Later, the privacy amplification was extended to the case that
an adversary have information encoded into a state of a quantum system
\cite{christandl:04, konig:05, renner:05c, renner:05b}).
Most of the following results can be found in \cite[Sections 3 and
5]{renner:05b},
but lemmas without citations are additionally proved in this
paper. We need Lemma \ref{lemma-weak-monotonicity}
to apply the results in \cite{renner:05b} to our proposed two-way
QKD protocol (QKD protocol with our proposed IR protocol). 
More specifically,
Eq.~(3.22) in \cite[Theorem 3.2.12]{renner:05b} plays an important
role to show a statement similar as Corollary
\ref{lemma-weak-chain-rule} in the case of one-way QKD protocol
(QKD protocol with one-way IR protocol).
However, the condition of Eq.~(3.22) in \cite[Theorem
3.2.12]{renner:05b}
is too restricted, and cannot be applied to our protocol.
Thus, we showed Corollary \ref{lemma-weak-chain-rule}
via Lemma \ref{lemma-weak-monotonicity}.
Lemmas \ref{neighbor-purification} and 
\ref{neighbor-purification-classical} are needed
 to prove Lemma \ref{lemma-weak-monotonicity}.
Lemmas \ref{lemma-matrix-inequality-1}--\ref{lemma-continuity}
are implicitly used in \cite{renner:05b} without proof,
which are also used in our proof in Section \ref{sec-proof-of-theorem}.

\subsection{Min- and Max- Entropy}

The (smooth) min-entropy and (smooth) max-entropy are formally defined as follows.
\begin{definition}
\cite[Definition 3.1.1]{renner:05b}
Let $\rho_{AB} \in {\cal P}({\cal H}_A \otimes {\cal H}_B)$ and
$\sigma_B \in {\cal P}({\cal H}_B)$. The min-entropy of 
$\rho_{AB}$ relative to $\sigma_B$ is defined by
\begin{eqnarray*}
H_{\min}(\rho_{AB}|\sigma_B) := - \log \lambda,
\end{eqnarray*}
where $\lambda$ is the minimum real number such that 
$\lambda \cdot \rom{id}_A \otimes \sigma_B - \rho_{AB} \ge 0$,
where $\rom{id}_A$ is the identity operator on ${\cal H}_A$.
When the condition $\rom{supp}(\rho_B) \subset \rom{supp}(\sigma_B)$ does not hold,
there is no $\lambda$ satisfying the condition 
$\lambda \cdot \rom{id}_A \otimes \sigma_B - \rho_{AB} \ge 0$, thus we define
$H_{\min}(\rho_{AB}|\sigma_B) := - \infty$.

The max-entropy of $\rho_{AB}$ relative to $\sigma_B$ is defined by
\begin{eqnarray*}
H_{\max}(\rho_{AB}|\sigma_B) := \log \rom{Tr} \left( (\rom{id}_A \otimes \sigma_B) 
\rho_{AB}^0 \right),
\end{eqnarray*}
where $\rho_{AB}^0$ denotes the projector onto the support of $\rho_{AB}$.

The min-entropy and the max-entropy of $\rho_{AB}$ given ${\cal H}_B$ are defined by
\begin{eqnarray*}
H_{\min}(\rho_{AB}|B) &:=& \sup_{\sigma_B} H_{\min}(\rho_{AB}|\sigma_B) \\
H_{\max}(\rho_{AB}|B) &:=& \sup_{\sigma_B} H_{\max}(\rho_{AB}|\sigma_B),
\end{eqnarray*}
where the supremum ranges over all $\sigma_B \in {\cal P}({\cal H}_B)$
with $\rom{Tr} \sigma_B = 1$.
\end{definition}

When ${\cal H}_B$ is the trivial space $\mathbb{C}$, the min-entropy and
the max-entropy of $\rho_A$ is 
\begin{eqnarray*}
H_{\min}(\rho_A) &=& - \log \lambda_{\max}( \rho_A) \\
H_{\max}(\rho_A) &=& \log \rom{rank}(\rho_A),
\end{eqnarray*}
where $\lambda_{\max}(\cdot)$ denotes the maximum eigenvalue of the
argument.

\begin{definition}
\cite[Definitions 3.2.1 and 3.2.2]{renner:05b}
Let $\rho_{AB} \in {\cal P}({\cal H}_A \otimes {\cal H}_B)$,
$\sigma_B \in {\cal P}({\cal H}_B)$, and $\varepsilon \ge 0$.
The $\varepsilon$-smooth min-entropy and the $\varepsilon$-smooth max-entropy
of $\rho_{AB}$ relative to $\sigma_B$ are defined by
\begin{eqnarray*}
H_{\min}^{\varepsilon}(\rho_{AB}|\sigma_B) &:=&
\sup_{\overline{\rho}_{AB}} H_{\min}(\overline{\rho}_{AB}|\sigma_B) \\
H_{\max}^{\varepsilon}(\rho_{AB}|\sigma_B) &:=&
\inf_{\overline{\rho}_{AB}} H_{\max}(\overline{\rho}_{AB}|\sigma_B),
\end{eqnarray*}
where the supremum and infimum ranges over the set
${\cal B}^{\varepsilon}(\rho_{AB})$ of all operators
$\overline{\rho}_{AB} \in {\cal P}({\cal H}_A \otimes {\cal H}_B)$ such
that $\| \overline{\rho}_{AB} - \rho_{AB} \| \le (\rom{Tr} \rho_{AB}) \varepsilon$.

The conditional 
$\varepsilon$-smooth min-entropy and the $\varepsilon$-smooth max-entropy
of $\rho_{AB}$ given ${\cal H}_B$ are defined by
\begin{eqnarray*}
H_{\min}^{\varepsilon}(\rho_{AB}|B) &:=& \sup_{\sigma_B} H_{\min}^{\varepsilon}(\rho_{AB}|\sigma_B) \\
H_{\max}^{\varepsilon}(\rho_{AB}|B) &:=& \sup_{\sigma_B} H_{\max}^{\varepsilon}(\rho_{AB}|\sigma_B),
\end{eqnarray*}
where the supremum ranges over all $\sigma_B \in {\cal P}({\cal H}_B)$
with $\rom{Tr} \sigma_B = 1$.
\end{definition}

The following lemma is a kind of chain rule for the smooth Min-entropy.
\begin{lemma}
     \label{lemma-chain-rule}
\cite[Theorem 3.2.12]{renner:05b}
For a tripartite operator 
$\rho_{ABC} \in {\cal P}({\cal H}_A \otimes {\cal H}_B \otimes {\cal
 H}_C)$, we have
\begin{eqnarray}
   \label{chain-rule}
H_{\min}^{\varepsilon}(\rho_{ABC}|C)
\le H_{\min}^{\varepsilon}(\rho_{ABC}|BC) + H_{\max}(\rho_B).
\end{eqnarray}
\end{lemma}

The following lemma states that removing the classical system 
only decreases the Min-entropy.
\begin{lemma}
\label{monotonicity-of-min-entropy}
\cite[Lemma 3.1.9]{renner:05b}
(monotonicity of min-entropy)
Let $\rho_{XBC} \in {\cal P}({\cal H}_X \otimes {\cal H}_B \otimes {\cal H}_C)$ be
classical on ${\cal H}_X$, and let $\sigma_C \in {\cal P}({\cal H}_C)$.
Then, we have
\begin{eqnarray*}
H_{\min}(\rho_{XBC}|\sigma_C) \ge H_{\min}(\rho_{BC}|\sigma_C).
\end{eqnarray*}
\end{lemma}
In order to extend Lemma \ref{monotonicity-of-min-entropy} to the
smooth min-entropy, we need Lemmas \ref{neighbor-purification} 
and \ref{neighbor-purification-classical}.
\begin{lemma}
\label{neighbor-purification}
Let $\rho_{AB} \in {\cal P}({\cal H}_A \otimes {\cal H}_B)$ be a density operator.
For $\varepsilon \ge 0$, let 
$\hat{\rho}_B \in {\cal B}^\varepsilon(\rho_B)$. Then, there exists
a operator $\hat{\rho}_{AB} \in {\cal B}^{\bar{\varepsilon}}(\rho_{AB})$ such
that $\rom{Tr}_A[\hat{\rho}_{AB}] = \hat{\rho}_B$,
where $\bar{\varepsilon} := \sqrt{8 \varepsilon}$.
\end{lemma}
\begin{proof}
Since $\hat{\rho}_B \in {\cal B}^\varepsilon(\rho_B)$, we have
\begin{eqnarray*}
\| \hat{\rho}_B \| \ge \| \rho_B \| - \| \rho_B - \hat{\rho}_B \| \ge 1 - \varepsilon.
\end{eqnarray*}
Then, from Lemma \ref{lower-bound-by-fidelity}, we have
\begin{eqnarray*}
F(\rho_B, \hat{\rho}_B) &\ge& \frac{1}{2} (
\rom{Tr}\rho_B + \rom{Tr} \hat{\rho}_B - \| \rho_B - \hat{\rho}_B \| ) \\
&\ge& 1 - \varepsilon.
\end{eqnarray*}
Let $\ket{\Psi} \in {\cal H}_R \otimes {\cal H}_A \otimes {\cal H}_B$ be a purification
of $\rho_{AB}$. Then, from Theorem \ref{uhlman-theorem}, there exists a purification
$\ket{\Phi} \in {\cal H}_R \otimes {\cal H}_A \otimes {\cal H}_B$ of
$\hat{\rho}_B$ such that
\begin{eqnarray*}
F(\ket{\Psi}, \ket{\Phi}) = F(\rho_B, \hat{\rho}_B) \ge 1 - \varepsilon.
\end{eqnarray*}
By noting that $F(\ket{\Psi}, \ket{\Phi})^2 \ge 1 - 2\varepsilon$, from
Lemma \ref{upper-bound-by-fidelity}, we have
\begin{eqnarray*}
\| \ket{\Psi}\bra{\Psi} - \ket{\Phi}\bra{\Phi} \| \le
\sqrt{8 \varepsilon}.
\end{eqnarray*}
Let $\hat{\rho}_{AB} := \rom{Tr}_R[ \ket{\Phi}\bra{\Phi} ]$.
Then, since the trace distance does not increase by the partial trace, we have
\begin{eqnarray*}
\| \rho_{AB} - \hat{\rho}_{AB} \| \le \sqrt{8 \varepsilon}.
\end{eqnarray*}
\end{proof}
\begin{remark}
In Lemma \ref{neighbor-purification}, if the density operator
$\rho_{AB}$ is classical with respect to both systems 
${\cal H}_A  \otimes {\cal H}_B$, then we can easily replace
$\bar{\varepsilon}$ by $\varepsilon$.
Then, $\bar{\varepsilon}$ in Lemma \ref{neighbor-purification-classical},
\ref{lemma-weak-monotonicity} and Corollary
\ref{lemma-weak-chain-rule} can also be replaced by
$\varepsilon$.
\end{remark}
\begin{lemma}
\label{neighbor-purification-classical}
Let $\rho_{XB} \in {\cal P}({\cal H}_X \otimes {\cal H}_B)$ be a density operator
that is classical on ${\cal H}_X$.
For $\varepsilon \ge 0$, let $\hat{\rho}_B \in {\cal B}^\varepsilon(\rho_B)$.
Then, there exists a operator $\hat{\rho}_{XB} \in {\cal B}^{\bar{\varepsilon}}(\rho_{XB})$
such that $\rom{Tr}_X[ \hat{\rho}_{XB} ] = \hat{\rho}_B$ and $\hat{\rho}_{XB}$ is classical on
${\cal H}_X$, where $\bar{\varepsilon} := \sqrt{8 \varepsilon}$.
\end{lemma}
\begin{proof}
From Lemma \ref{neighbor-purification}, there exists a operator 
$\rho^\prime_{XB} \in {\cal B}^{\bar{\varepsilon}}(\rho_{XB})$ such
that $\rom{Tr}_X[ \rho_{XB}^\prime ] = \hat{\rho}_B$.
Let ${\cal E}_X$ be a projection measurement CP map on ${\cal H}_X$, i.e.,
\begin{eqnarray*}
{\cal E}_X(\rho) := \sum_{x \in {\cal X}} \ket{x}\bra{x} \rho \ket{x}\bra{x},
\end{eqnarray*}
where $\{ \ket{x} \}_{x \in {\cal X}}$ is an orthonormal basis of ${\cal H}_X$.
Let $\hat{\rho}_{XB} := ({\cal E}_X \otimes \rom{id}_B)(\rho^\prime_{XB})$.
Then, since the trace distance does not increase by the CP map, and 
$({\cal E}_X \otimes \rom{id}_B)(\rho_{XB}) = \rho_{XB}$, we have
\begin{eqnarray*}
\lefteqn{
\| \hat{\rho}_{XB} - \rho_{XB} \| } \\
&=&
\| ({\cal E}_X \otimes \rom{id}_B)(\rho_{XB}^\prime) -
({\cal E}_X \otimes \rom{id}_B)(\rho_{XB}) \| \\
&\le&
\| \rho_{XB}^\prime - \rho_{XB}
\| \\
&\le& \bar{\varepsilon}.
\end{eqnarray*}
Furthermore, we have
$\rom{Tr}_X [ \hat{\rho}_{XB} ] = \rom{Tr}_X [ \rho_{XB}^\prime ] = \hat{\rho}_B$,
and $\hat{\rho}_{XB}$ is classical on ${\cal H}_X$.
\end{proof}

The following lemma states that
the monotonicity of the min-entropy (Lemma \ref{monotonicity-of-min-entropy})
can be extended to the smooth min-entropy 
by adjusting the smoothness $\varepsilon$. 
\begin{lemma}
\label{lemma-weak-monotonicity}
Let $\rho_{XBC} \in {\cal P}({\cal H}_X \otimes {\cal H}_B \otimes {\cal H}_C)$
be a density operator that is classical on ${\cal H}_X$. 
Then, for any $\varepsilon \ge 0$, we have
\begin{eqnarray*}
H_{\min}^{\bar{\varepsilon}}(\rho_{XBC} | C) \ge 
H_{\min}^{\varepsilon}(\rho_{BC} | C),
\end{eqnarray*}
where $\bar{\varepsilon} := \sqrt{8 \varepsilon}$.
\end{lemma}
\begin{proof}
We will prove that
\begin{eqnarray*}
H_{\min}^{\bar{\varepsilon}}(\rho_{XBC}|\sigma_C) \ge 
H_{\min}^{\varepsilon}(\rho_{BC} | \sigma_C)
\end{eqnarray*}
holds for any $\sigma_C \in {\cal P}({\cal H}_C)$ with $\rom{Tr} \sigma_C = 1$.
From the definition of the smooth min-entropy,
for any $\nu > 0$, there exists $\hat{\rho}_{BC} \in {\cal B}^{\varepsilon}(\rho_{BC})$
such that
\begin{eqnarray}
\label{eq-proof-of-weak-monotonicity-1}
H_{\min}(\hat{\rho}_{BC} | \sigma_C) \ge H_{\min}^{\varepsilon}(\rho_{BC}|\sigma_C) -\nu.
\end{eqnarray}
From Lemma \ref{neighbor-purification-classical}, there exists a operator
$\hat{\rho}_{XBC} \in {\cal B}^{\bar{\varepsilon}}(\rho_{XBC})$ such that
$\rom{Tr}_X [ \hat{\rho}_{XBC} ] = \hat{\rho}_{BC}$, and
$\hat{\rho}_{XBC}$ is classical on ${\cal H}_X$.
Then, from Lemma \ref{monotonicity-of-min-entropy}, we have
\begin{eqnarray}
\label{eq-proof-of-weak-monotonicity-2}
H_{\min}(\hat{\rho}_{XBC}|\sigma_C) \ge H_{\min}(\hat{\rho}_{BC}|\sigma_C).
\end{eqnarray}
Furthermore, from the definition of smooth min-entropy, we have
\begin{eqnarray}
\label{eq-proof-of-weak-monotonicity-3}
H_{\min}^{\bar{\varepsilon}}(\rho_{XBC} | \sigma_C) \ge 
H_{\min}(\hat{\rho}_{XBC} | \sigma_C).
\end{eqnarray}
Since $\nu > 0$ is arbitrary, combining 
Eqs.~(\ref{eq-proof-of-weak-monotonicity-1})--(\ref{eq-proof-of-weak-monotonicity-3}), we 
have the assertion of the lemma.
\end{proof}

Combining Eq.~(\ref{chain-rule}) of Lemma \ref{lemma-chain-rule}
and Lemma \ref{lemma-weak-monotonicity}, we have the following corollary,
which states that the condition decreases 
the smooth min-entropy by at most the amount of
the max-entropy of the condition, and 
plays an important role to prove security of QKD protocols.
\begin{corollary}
\label{lemma-weak-chain-rule}
Let $\rho_{XBC} \in {\cal P}({\cal H}_X \otimes {\cal H}_B \otimes {\cal H}_C)$
be a density operator that is classical on ${\cal H}_X$. 
Then, for any $\varepsilon \ge 0$, we have
\begin{eqnarray*}
H_{\min}^{\bar{\varepsilon}}(\rho_{XBC} | XC) \ge 
H_{\min}^{\varepsilon}(\rho_{BC} | C) - H_{\max}(\rho_X),
\end{eqnarray*}
where $\bar{\varepsilon} := \sqrt{8 \varepsilon}$.
\end{corollary}

The following lemmas are also used in Section \ref{sec-proof-of-theorem}.

\begin{lemma}
\cite[Theorem 3.2.12]{renner:05b}
	\label{lemma-properties-of-min-max-entropy}
The following inequalities hold:
\begin{itemize}
\item Strong sub-additivity:
\begin{eqnarray}
\label{strong-subadditivity}
H_{\min}^{\varepsilon}(\rho_{ABC}|BC) \le H_{\min}^{\varepsilon}(\rho_{AB}|B)
\end{eqnarray}
for $\rho_{ABC} \in {\cal P}({\cal H}_A \otimes {\cal H}_B \otimes {\cal H}_C)$.

\item Conditioning on classical information:
\begin{eqnarray}
\label{conditioning-on-classical-information-2}
H_{\min}^{\varepsilon}(\rho_{ABZ}|BZ) \ge \min_{z \in {\cal Z}}
H_{\min}^{\varepsilon}(\rho_{AB}^z|B)
\end{eqnarray}
for $\rho_{ABZ} \in {\cal P}({\cal H}_A \otimes {\cal H}_B \otimes {\cal H}_Z)$
normalized and
classical on ${\cal H}_Z$, and 
for conditional operators
$\rho_{AB}^z \in {\cal P}({\cal H}_A \otimes {\cal H}_B)$
and 
$\rho_B^z \in {\cal P}({\cal H}_B)$.
\end{itemize}
\end{lemma}

In order to prove that removing the (not necessarily classical) system
increases the min-entropy at most the max entropy of the removed system
(Lemma \ref{lemma-chain-1}), we need the following lemma.
\begin{lemma}
	\label{lemma-matrix-inequality-1}
Let $\rho_{AB} \in {\cal P}({\cal H}_A \otimes {\cal H}_B)$ 
be a density operator, and
let $r_A := \rom{rank}(\rho_A)$.
Then, we have
\begin{eqnarray*}
r_A \rom{id}_A \otimes \rho_B - \rho_{AB} \ge 0,
\end{eqnarray*}
where $\rom{id}_A$ is the identity operator on ${\cal H}_A$.
\end{lemma}
\begin{proof}
First, we prove the assertion for pure state $\rho_{AB} = \ket{\Psi}\bra{\Psi}$.
Let 
\begin{eqnarray}
     \label{schidt-decomposition}
\ket{\Psi} = \sum_{i=1}^{r_A} \sqrt{\alpha_i} \ket{\phi_i} \otimes
\ket{\psi_i}
\end{eqnarray}
be a Schmidt decomposition of $\ket{\Psi}$.
Let $\{ \ket{\phi_i} \}_{i=1}^{d_A}$ and $\{ \ket{\psi_i} \}_{i=1}^{d_B}$
be orthonormal bases of ${\cal H}_A$ and ${\cal H}_B$ that are 
extensions of vectors in Eq.~(\ref{schidt-decomposition}).
For any vector $\ket{\Phi} \in {\cal H}_A \otimes {\cal H}_B$, we can write
\begin{eqnarray*}
\ket{\Phi} = \sum_{i=1}^{d_B} \beta_i \ket{\hat{\phi}_i} \otimes 
\ket{\psi_i},
\end{eqnarray*}
where $\{ \ket{\hat{\phi}_i} \}_{i = 1}^{d_B}$ is normalized but not
necessarily orthogonal. Then, we have
\begin{eqnarray*}
\bra{\Phi} \rho_{AB} \ket{\Phi} = |\bra{\Psi}\ket{\Phi} |^2
&=& \left| \sum_{i = 1}^{r_A} \sqrt{\alpha_i} \beta_i 
\bra{\phi_i}\ket{\hat{\phi}_i} \right|^2 \\
&\le& \left| \sum_{i =1}^{r_A} \sqrt{\alpha_i} |\beta_i| \right|^2
\end{eqnarray*} 
and 
\begin{eqnarray*}
\bra{\Phi} ( r_A \rom{id}_A \otimes \rho_B) \ket{\Phi}
&=& r_A \| \sum_{i = 1}^{r_A} \sqrt{\alpha_i} \beta_i \ket{\hat{\phi}_i} 
\otimes \ket{\psi_i} \|^2 \\
&=& r_A \sum_{i=1}^{r_A} \alpha_i |\beta_i|^2.
\end{eqnarray*}
Using the Cauchy-Schwartz inequality for two vectors
$(1,\ldots, 1)$ and 
$(\sqrt{\alpha_1} |\beta_1|$, $\ldots$, $\sqrt{\alpha_{r_A}} |\beta_{r_A}|)$, 
we have
\begin{eqnarray*}
\bra{\Phi} \rho_{AB} \ket{\Phi}
\le \left| \sum_{i =1}^{r_A} \sqrt{\alpha_i} |\beta_i| \right|^2 
&\le& r_A \sum_{i=1}^{r_A} \alpha_i |\beta_i|^2 \\
&=& \bra{\Phi} ( r_A \rom{id}_A \otimes \rho_B) \ket{\Phi}.
\end{eqnarray*}
Thus, the assertion holds for a pure state $\rho_{AB} = \ket{\Psi}\bra{\Psi}$.
For a mixed state $\rho_{AB}$, let 
$\rho_{AB} = \sum_{i=1}^{m} p_i \ket{\Psi_i}\bra{\Psi_i}$ be 
an eigenvalue decomposition. Let $\rho_B^i = \rom{Tr}_A \ket{\Psi_i}\bra{\Psi_i}$.
Noting that 
$\rom{rank}(\rom{Tr}_B \ket{\Psi_i}\bra{\Psi_i}) \le \rom{rank}(\rom{Tr}_B \rho_{AB})
= r_A$ for all $1 \le i \le m$,
we have
\begin{eqnarray*}
r_A \rom{id}_A \otimes \rho_{B} - \rho_{AB} 
= \sum_{i=1}^m p_i
( r_A \rom{id}_A \otimes \rho_B^i - \ket{\Psi_i}\bra{\Psi_i} ) \ge 0.
\end{eqnarray*}
\end{proof}
%
\begin{lemma}
	\label{lemma-chain-1}
Let $\rho_{ABC} \in {\cal P}({\cal H}_A \otimes {\cal H}_B \otimes {\cal H}_C)$ and
$\sigma_C \in {\cal P}({\cal H}_C)$. Then
\begin{eqnarray*}
H_{\min}(\rho_{ABC}|\sigma_C) \ge H_{\min}(\rho_{BC}|\sigma_C) -
H_{\max}(\rho_A).
\end{eqnarray*}
\end{lemma}
\begin{proof}
Let $\lambda$ is such that
$H_{\min}(\rho_{BC}|\sigma_C) = - \log \lambda$, i.e., $\lambda$ is the minimum number 
satisfying
\begin{eqnarray*}
\lambda \rom{id}_B \otimes \sigma_C - \rho_{BC} \ge 0.
\end{eqnarray*}
Let $r_A := \rom{rank}(\rho_A)$. Then, we want to show that
\begin{eqnarray*}
H_{\min}(\rho_{ABC}|\sigma_C) \ge - \log \lambda - \log r_A = - \log r_A \lambda,
\end{eqnarray*}
i.e.,
$r_A \lambda \rom{id}_{AB} \otimes \sigma_C - \rho_{ABC} \ge 0$.
From Lemma \ref{lemma-matrix-inequality-1}, we have
\begin{eqnarray*}
\lefteqn{ r_A \lambda \rom{id}_{AB} \otimes \sigma_C - \rho_{ABC} } \\
&\ge& r_A \lambda \rom{id}_{AB} \otimes \sigma_C - r_A \rom{id}_A \otimes \rho_{BC} \\
&=& r_A \rom{id}_A \otimes ( \lambda \rom{id}_B \otimes \sigma_C - \rho_{BC} ) \ge 0.
\end{eqnarray*}
\end{proof}

The following lemma states that Lemma \ref{lemma-chain-1}
can be extended to the smooth Min-entropy by
adjusting the smoothness $\varepsilon$.
\begin{lemma}
\label{lemma-chain-2}
Let $\varepsilon \ge 0$ and 
$\rho_{ABC} \in {\cal P}({\cal H}_A \otimes {\cal H}_B \otimes {\cal H}_C)$
be a density operator.
Then, we have
\begin{eqnarray*}
H_{\min}^{\bar{\varepsilon}}(\rho_{ABC}|C) \ge H_{\min}^\varepsilon(\rho_{BC}|C) 
- \log \dim {\cal H}_A,
\end{eqnarray*}
where $\bar{\varepsilon} := \sqrt{8 \varepsilon}$.
\end{lemma}
\begin{proof}
We will prove that
\begin{eqnarray*}
H_{\min}^{\bar{\varepsilon}}(\rho_{ABC}|\sigma_C) \ge
H_{\min}^{\varepsilon}(\rho_{BC}|\sigma_C) - \log \dim {\cal H}_A
\end{eqnarray*}
holds for any $\sigma_C \in {\cal P}({\cal H}_C)$ with $\rom{Tr} \sigma_C = 1$.
For any $\nu > 0$, there exists $\hat{\rho}_{BC} \in {\cal B}^{\varepsilon}(\rho_{BC})$
such that
\begin{eqnarray}
\label{eq-proof-of-chain-2-1}
H_{\min}(\hat{\rho}_{BC}|\sigma_C) \ge H_{\min}^{\varepsilon}(\rho_{BC}|\sigma_C) - \nu.
\end{eqnarray}
From Lemma \ref{neighbor-purification}, there exists a operator
$\hat{\rho}_{ABC} \in {\cal B}^{\bar{\varepsilon}}(\rho_{ABC})$
such that $\rom{Tr}_A[ \hat{\rho}_{ABC}] = \hat{\rho}_{BC}$.
Then from Lemma \ref{lemma-chain-1}, we have
\begin{eqnarray}
\label{eq-proof-of-chain-2-2}
H_{\min}(\hat{\rho}_{ABC}|\sigma_C) \ge
H_{\min}(\hat{\rho}_{BC} | \sigma_C) - \log \dim {\cal H}_A.
\end{eqnarray}
Furthermore, from the definition of the smooth-min-entropy, we have
\begin{eqnarray}
\label{eq-proof-of-chain-2-3}
H_{\min}^{\bar{\varepsilon}}(\rho_{ABC}|\sigma_C) 
\ge H_{\min}(\hat{\rho}_{ABC}|\sigma_C).
\end{eqnarray}
Since $\nu > 0$ is arbitrary, combining 
Eqs.~(\ref{eq-proof-of-chain-2-1})--(\ref{eq-proof-of-chain-2-3}),
we have the assertion of the lemma.
\end{proof}

\begin{lemma}
	\label{lemma-continuity}
For density operators
$\rho_{AB}, \overline{\rho}_{AB} \in {\cal P}({\cal H}_A \otimes {\cal H}_B)$ such 
that $\| \rho_{AB} - \overline{\rho}_{AB} \| \le \varepsilon^\prime$, we have
\begin{eqnarray*}
H_{\min}^{\varepsilon}(\rho_{AB}|B) \le 
H_{\min}^{\varepsilon + \varepsilon^\prime}(\overline{\rho}_{AB}|B) 
\end{eqnarray*}
\end{lemma}
\begin{proof}
For all $\hat{\rho}_{AB} \in {\cal B}^{\varepsilon}(\rho_{AB})$, by the triangle inequality, we have
\begin{eqnarray*}
\| \overline{\rho}_{AB} - \hat{\rho}_{AB} \| \le
\| \rho_{AB} - \hat{\rho}_{AB} \| +  \| \overline{\rho}_{AB} - \rho_{AB} \|
\le \varepsilon + \varepsilon^\prime.
\end{eqnarray*}
Thus, we have $\hat{\rho}_{AB} \in {\cal B}^{\varepsilon + \varepsilon^\prime}(\overline{\rho}_{AB})$,
and 
\begin{eqnarray*}
H_{\min}^{\varepsilon}(\rho_{AB}|\sigma_B) \le 
H_{\min}^{\varepsilon + \varepsilon^\prime}(\overline{\rho}_{AB}|\sigma_B) 
\end{eqnarray*}
for all $\sigma_B \in {\cal P}({\cal H}_B)$. Then we have the assertion of the lemma.
\end{proof}

\subsection{Privacy amplification}

The following definition is used to state the 
security of the distilled key by the privacy amplification.

\begin{definition}
\cite[Definition 5.2.1]{renner:05b}
\label{definition-distance-from-uniform}
Let $\rho_{AB} \in {\cal P}({\cal H}_A \otimes {\cal H}_B)$.
Then the trace distance from the uniform of $\rho_{AB}$ given $B$ is
defined by
\begin{eqnarray*}
d(\rho_{AB}|B) := \| \rho_{AB} - \rho_A^{\rom{mix}} \otimes \rho_B \|,
\end{eqnarray*}
where $\rho_A^{\rom{mix}} := \frac{1}{\dim {\cal H}_A} \rom{id}_A$ is the fully
mixed state on ${\cal H}_A$ and 
$\rho_B := \rom{Tr}_A[\rho_{AB}]$.
\end{definition}
\begin{definition}
\cite{carter:79}
\label{definition-two-universal-hash}
Let ${\cal F}$ be a family of hash functions from ${\cal X}$ to
${\cal Z}$, and let $P_F$ be the uniform probability distribution on ${\cal F}$.
The family ${\cal F}$ is called two-universal if
$\Pr \{ f(x) = f(x^\prime) \} \le \frac{1}{|{\cal Z}|}$ for
any distinct $x, x^\prime \in {\cal X}$.
\end{definition}

Consider an operator $\rho_{XB} \in {\cal P}({\cal H}_X \otimes {\cal H}_B)$ that is classical with respect to
an orthonormal basis $\{ \ket{x} \}_{x \in {\cal X}}$ of ${\cal H}_X$,
and assume that $f$ is a function from ${\cal X}$ to ${\cal Z}$.
The operator describing the classical function output together
with the quantum system ${\cal H}_B$ is then given by
\begin{eqnarray}
\label{eq-state-key}
\rho_{f(X)B} := \sum_{z \in {\cal Z}} \ket{z}\bra{z} \otimes \rho_B^z
~\mbox{for } \rho_B^z := \sum_{x \in f^{-1}(z)} \rho_B^x,
\end{eqnarray}
where $\{ \ket{z} \}_{z \in {\cal Z}}$ is an orthonormal basis of ${\cal H}_Z$.

Assume now that the function $f$ is randomly chosen from a family ${\cal
F}$
of function according to the uniform probability distribution $P_F$.
Then the output $f(x)$, the state of the quantum system, and the
choice of the function $f$ is  described by the operator
\begin{eqnarray}
\label{state-distilled-key}
\rho_{F(X)BF} := \sum_{f \in {\cal F}} P_F(f) \rho_{f(X)B} \otimes
\ket{f}\bra{f}
\end{eqnarray}
on ${\cal H}_Z \otimes {\cal H}_B \otimes {\cal H}_F$, where ${\cal H}_F$
is a Hilbert space with orthonormal basis $\{ \ket{f} \}_{f \in {\cal F}}$.
The system ${\cal H}_Z$ describes the distilled key, and 
the system ${\cal H}_B$ and  ${\cal H}_F$ describe the information
which an adversary Eve can access. 
The following lemma states that the length of securely distillable key
is given by the conditional smooth min-entropy
$H_{\min}^{\varepsilon}(\rho_{XB}|B)$.

\begin{lemma}
\cite[Corollary 5.6.1]{renner:05b}
      \label{lemma-privacy-amplification}
Let $\rho_{XB} \in {\cal P}({\cal H}_X \otimes {\cal H}_B)$ be a density
operator which is classical with respect to an orthonormal basis 
$\{ \ket{x} \}_{x \in {\cal X}}$ of ${\cal H}_X$.
Let ${\cal F}$ be a two-universal family of hash functions from
${\cal X}$ to $\{ 0,1\}^\ell$, and let $\varepsilon \ge 0$.
Then we have
\begin{eqnarray*}
d(\rho_{F(X)BF}|BF) \le 2 \varepsilon + 
2^{-\frac{1}{2} ( H_{\min}^{\varepsilon}(\rho_{XB}|B) - \ell)}
\end{eqnarray*}
for $\rho_{F(X)BF} \in {\cal P}({\cal H}_Z \otimes {\cal H}_B \otimes {\cal H}_F)$
defined by Eq.~(\ref{state-distilled-key}).
\end{lemma}

\section{Symmetric states}
\label{sec-symmetric-states}

In this section, we review the property of symmetric states
and the de Finetti style representation theorem \cite{renner:05b, renner:07}.
For more detail, refer to \cite[Section 4]{renner:05b}.

Let ${\cal H}$ be a Hilbert space and let ${\cal S}_n$ be the set of 
permutations on $\{1,\ldots,n\}$.
For any $\pi \in {\cal S}_n$, we denote by the same letter $\pi$ the
unitary operation on ${\cal H}^{\otimes n}$ which permutes the $n$
subsystems, that is,
\begin{eqnarray*}
\pi(\ket{\theta_1} \otimes \cdots \otimes \ket{\theta_n}) :=
\ket{\theta_{\pi^{-1}(1)}} \otimes \cdots \otimes \ket{\theta_{\pi^{-1}(n)}},
\end{eqnarray*}
for any $\ket{\theta_1}, \ldots, \ket{\theta_n} \in {\cal H}$.
\begin{definition}
\cite[Definition 4.1.1]{renner:05b}
     \label{definition-symmetric-subspace}
The symmetric subspace $\mbox{Sym}({\cal H}^{\otimes n})$ of ${\cal H}^{\otimes n}$
is the subspace of ${\cal H}^{\otimes n}$ spanned by all vectors
which are invariant under permutations of the subsystems, that is,
\begin{eqnarray*}
\mbox{Sym}({\cal H}^{\otimes n}) := \{
\ket{\Psi} \in {\cal H}^{\otimes n} \mid \pi \ket{\Psi} =
\ket{\Psi},~\forall \pi \in {\cal S}_n \}.
\end{eqnarray*}
\end{definition}

\begin{definition}
\cite[Definition 4.1.4]{renner:05b}
Let $\ket{\theta} \in {\cal H}$ be fixed, and let $0 \le m \le n$.
We denote by ${\cal V}({\cal H}^{\otimes n}, \ket{\theta}^{\otimes m})$ the
set of vectors $\ket{\Psi} \in {\cal H}^{\otimes n}$ which,
after some reordering of the subsystems, are of the form
$\ket{\theta}^{\otimes m} \otimes \ket{\tilde{\Psi}}$, that is,
\begin{eqnarray*}
\lefteqn{ 
{\cal V}({\cal H}^{\otimes n}, \ket{\theta}^{\otimes m}) } \\
&:=&
\{ \pi ( \ket{\theta}^{\otimes m} \otimes \ket{\tilde{\Psi}} ) \mid 
\pi \in {\cal S}_n,~ \ket{\tilde{\Psi}} \in {\cal H}^{\otimes n-m} \}.
\end{eqnarray*}
The symmetric subspace $\mbox{Sym}({\cal H}^{\otimes n}, \ket{\theta}^{\otimes m})$
of ${\cal H}^{\otimes n}$ along $\ket{\theta}^{\otimes m}$ is
\begin{eqnarray*}
\mbox{Sym}({\cal H}^{\otimes n}, \ket{\theta}^{\otimes m}) :=
\mbox{Sym}({\cal H}^{\otimes n}) \cap \mbox{span } 
{\cal V}({\cal H}^{\otimes n}, \ket{\theta}^{\otimes m}).
\end{eqnarray*}
\end{definition}
If $m \ll  n$, then we can consider that a state 
$\ket{\Psi} \in \mbox{Sym}({\cal H}^{\otimes n}, \ket{\theta}^{\otimes m})$ is 
almost the same as the product state $\ket{\theta}^{\otimes n}$.

The following lemma states that a permutation invariant mixed states have a 
purification in the symmetric space of a extended systems.
\begin{lemma}
\cite[Lemma 4.2.2]{renner:05b}
	\label{symmetric-purification}
Let $\rho_n \in {\cal P}({\cal H}^{\otimes n})$ be permutation-invariant.
Then, there exists a purification 
$\ket{\Psi} \in \mbox{Sym}(({\cal H} \otimes {\cal H})^{\otimes n})$
of $\rho_n$.
\end{lemma}

The following lemma states that a pure state on a symmetric space can be
approximated by a convex combination of pure states that are
close to a product state.
\begin{lemma}
\cite[Theorem 4.3.2]{renner:05b}
	\label{de-finneti-theorem}
Let $\rho_{n+k}$ be a pure state density operator on 
$\mbox{Sym}({\cal H}^{\otimes n+k})$ and let $0 \le r \le n$.
Then, there exists  a measure $\nu$ on 
${\cal S}_1({\cal H}) := \{ \ket{\theta} \in {\cal H} \mid \| \ket{\theta} \| = 1 \}$
and a pure
density operator $\bar{\rho}_n^{\ket{\theta}}$ on
$\mbox{Sym}({\cal H}^{\otimes n}, \ket{\theta}^{\otimes n-r})$ 
for each $\ket{\theta} \in {\cal S}_1({\cal H})$
such that
\begin{eqnarray*}
\left\| \rom{Tr}_k( \rho_{n+k}) -
\int_{{\cal S}_1({\cal H})} \bar{\rho}_n^{\ket{\theta}} \nu(\ket{\theta}) \right\|
\le 2 e^{- \frac{k(r+1)}{2(n+k)} + \frac{1}{2} \dim({\cal H}) \ln k},
\end{eqnarray*}
where the base of $\ln$ is $e$.
\end{lemma}

The following lemma states that the smooth min-entropy of a density operator
that is derived from a pure state on $\rom{Sym}({\cal H}^{\otimes n},\ket{\theta}^{\otimes n -m})$
can be approximated from below by the von Neumann entropy of a density operator that 
is derived from a product state $\ket{\theta}^{\otimes n}$.
\begin{lemma}
\cite[Theorem 4.4.1]{renner:05b}
	\label{lemma-min-entropy-of-symmetric-state}
Let $0 \le r \le \frac{1}{2} n$,  $\ket{\theta} \in {\cal H}$, and 
$\ket{\Psi} \in \mbox{Sym}({\cal H}^{\otimes n}, \ket{\theta}^{\otimes n -r})$
be normalized, and let ${\cal E}$ be a trace-preserving CP map  from
${\cal H}$ to ${\cal H}_X \otimes {\cal H}_B$ that is classical on
 ${\cal H}_X$, i.e., ${\cal E}(\rho)$ is a $\{cq\}$-state for any 
$\rho \in {\cal P}({\cal H})$.
Define $\rho_{X^n B^n} := {\cal E}^{\otimes n}(\ket{\Psi}\bra{\Psi})$
and $\sigma_{XB} := {\cal E}(\ket{\theta}\bra{\theta})$. Then, for any 
$\varepsilon > 0$,
\begin{eqnarray*}
\frac{1}{n} H_{\min}^{\varepsilon}(\rho_{X^n B^n} | B^n) 
\ge H(\sigma_{XB}) - H(\sigma_B) - \delta,
\end{eqnarray*}
where $\delta = ( \frac{5}{2} H_{\max}(\sigma_X) + 4) 
\sqrt{\frac{2 \log(4/\varepsilon)}{n} + h(r/n) }$.
\end{lemma}
\begin{lemma}
 \cite[Theorem 4.5.2]{renner:05b}
	\label{lemma-statistic-of-symmetric-state}
Let $0 \le r \le \frac{1}{2} n$, $\ket{\theta} \in {\cal H}$, and
$\ket{\Psi} \in \rom{Sym}({\cal H}^{\otimes n}, \ket{\theta}^{\otimes n - r})$
be normalized. Let ${\cal M} = \{ M_z \}_{z \in {\cal Z}}$ be a POVM
on ${\cal H}$, and let $P_Z$ be a probability distribution of the outcomes
of the measurement ${\cal M}$ applied to $\ket{\theta}\bra{\theta}$.
Then we have
\begin{eqnarray*}
\Pr_{\mbf{z}} \left[
\| P_{\mbf{z}} - P_Z \| > \alpha \right] \le \varepsilon
\end{eqnarray*}
for
\begin{eqnarray*}
\alpha := 2 \sqrt{\frac{\log (1/\varepsilon)}{n} + h(r/n) + \frac{|{\cal Z}|}{n}
\log ( \frac{n}{2} + 1)}
\end{eqnarray*}
where the probability is taken over the outcomes
$\mbf{z} = (z_1,\ldots,z_n)$ of the product measurement ${\cal M}^{\otimes n}$
applied to $\ket{\Psi}\bra{\Psi}$.
\end{lemma}
Lemma \ref{lemma-statistic-of-symmetric-state} states that
if the product measurement ${\cal M}^{\otimes n}$ is applied to $\ket{\Psi}\bra{\Psi}$,
then the probability such that  type 
$P_{\mbf{z}}$ of the outcomes deviates from 
the distribution $P_Z$ is small.

\section{Proof of Theorem 2}
\label{sec-proof-of-theorem}

In this section, we prove 
Theorem 2.
In Section \ref{subsec-security-against-known-adversary},
we first prove the security  agaist known adversary.
In Section \ref{subsec-fluctuation-of-error},
we analyze the parameter estimation protocol.
Then, using results in Sections 
\ref{subsec-security-against-known-adversary} and
\ref{subsec-fluctuation-of-error}, we prove 
Theorem 2.

\subsection{Security against known adversary}
\label{subsec-security-against-known-adversary}

In this section, we analyze a situation after the 
parameter estimation of the QKD protocols, i.e.,
we assume the following situation.
Alice and Bob have $2n$-bit binary sequences 
$(\mbf{x}, \mbf{y}) \in \mathbb{F}_2^{2n} \times \mathbb{F}_2^{2n}$
that is distributed according to a probability distribution
$P_{\mbf{X}\mbf{Y}}$, and Eve can access the quantum system
${\cal H}_E$ whose state $\rho_E^{\mbf{x},\mbf{y}}$ is correlated
to $(\mbf{x}, \mbf{y})$.
This situation can be described by a $\{ccq\}$-state
\begin{eqnarray*}
\rho_{\mbf{X}\mbf{Y}E} := \sum_{(\mbf{x}, \mbf{y})}
P_{\mbf{X}\mbf{Y}}(\mbf{x},\mbf{y}) 
\ket{\mbf{x},\mbf{y}}\bra{\mbf{x},\mbf{y}} \otimes 
\rho_E^{\mbf{x},\mbf{y}}.
\end{eqnarray*}
In the following, we follow the notations of 
Section 2 even though
the distribution $P_{\mbf{X}\mbf{Y}}$ is not necessarily
the product distribution $P_{XY}^{2n}$. 

In order to agree on a secure key pair 
$(S_A, S_B)$, Alice and Bob
perform the procedure as in 
Section 3.
Then, the situation after the IR protocol
and the privacy amplification can be described by a
$\{ccq\}$-state
\begin{eqnarray*}
\lefteqn{ \rho_{S_A S_B C E} 
:=  } \nonumber \\
&&  \sum_{(s_A, s_B)} P_{S_A S_B}(s_A, s_B)
\ket{s_A, s_B}\bra{s_A, s_B} \otimes \rho_{C E}^{s_A, s_B},
      \label{eq-distilled-key}
\end{eqnarray*}
where the classical system $C$ describes the exchanged
messages $(\mbf{T}_1, \hat{\mbf{T}}_2, \hat{\mbf{W}}_1)$
in the IR protocol and the choice $F$ of the hash function
in the PA protocol.
As in Section 3, 
the distilled key pair $(S_A, S_B)$ is said to be
$\varepsilon$-secure  if
\begin{eqnarray}
	\label{definition-of-e-secure}
\frac{1}{2} \| \rho_{S_A S_B E^\prime} - 
\rho_{S_A S_B}^{\rom{mix}} \otimes \rho_{E^\prime} \| \le \varepsilon,
\end{eqnarray}
where $\rho_{S_A S_B}^{\rom{mix}} := \sum_{s \in {\cal S}}
 \frac{1}{|{\cal S}|} \ket{s,s}\bra{s,s}$ is the uniformly distributed
key on ${\cal S}$.
The above security definition for the key distillation protocol can
be subdivided into two parts (see also \cite[Remark 6.1.3]{renner:05b}):
\begin{itemize}
\item The distilled key pair $(S_A, S_B)$ is $\varepsilon_c$-correct
if
\begin{eqnarray*} 
\sum_{s_A \neq s_B} P_{S_A S_B}(s_A, s_B) \le \varepsilon_c.
\end{eqnarray*}

\item The distilled key $S_A$ is $\varepsilon_s$-secret if
$\frac{1}{2} d(\rho_{S_A E^\prime} | E^\prime) \le \varepsilon_s$. 
\end{itemize}
In particular, if the distilled key $(s_A, s_B)$ is 
$\varepsilon_c$-correct and $\varepsilon_s$-secret,
then it is $(\varepsilon_c + \varepsilon_s)$-secure.

The following theorem gives the relation between
the security and the length of distilled key.
\begin{theorem}
     \label{theorem-key-distillation-from-raw-key}

Assume that  
Alice and Bob's bit sequence after
the IR protocol are identical to $\mbf{u}$ with
probability at least $1- \varepsilon_1$, i.e.,
\begin{eqnarray}
       \label{eq-probability-of-reconciled-key-agrees}
P_{\mbf{X}\mbf{Y}}(\{ (\mbf{x},\mbf{y}) : \hat{\mbf{u}} = \tilde{\mbf{u}} = \mbf{u} \})
\ge 1 - \varepsilon_1.
\end{eqnarray}
For a given number $R, R_0 > 0$, assume that
the rate of linear codes that are used in the IR protocol satisfy
$\frac{m}{n} \le R$ and $\frac{m_0}{n_0} \le R_0$ for all
$\underline{n}_0 \le n_0 \le \overline{n}_0$.
Furthermore assume that the length $\ell$ of the distilled key by the 
privacy amplification satisfies 
\begin{eqnarray}
\lefteqn{ \ell \le  \max[
H_{\min}^{\varepsilon} (\rho_{\mbf{U} \mbf{W}_1 E} | \mbf{W}_1 E)
- nR - \overline{n}_0 R_0, } \nonumber \\
&&  H_{\min}^{\varepsilon}( \rho_{\mbf{U} 
 \mbf{W}_1 \mbf{U}_1 E} | \mbf{W}_1 \mbf{U}_1 E) - \overline{n}_0 R_0
] - \log(1/8\varepsilon),
\label{eq-the-secure-key-rate}
\end{eqnarray}
where $\rho_{\mbf{U} \mbf{W}_1 E}$ and
$ \rho_{\mbf{U} \mbf{W}_1 \mbf{U}_1 E}$ are derived from
$\rho_{\mbf{X}\mbf{Y}E}$ by using the functions $\xi_1$
and $\xi_2$ in the same way as in 
Section 2.
Then the distilled key pair $(S_A, S_B)$ is 
$(\bar{\varepsilon} + 3 \varepsilon_1)$-secure, where
$\bar{\varepsilon} := \frac{3}{2} \sqrt{8 \varepsilon}$.
\end{theorem}
\begin{proof}
First, we will prove that the dummy key $S := f(\mbf{U})$ is 
$\bar{\varepsilon}$-secret under the condition that
Eve can access $(\mbf{W}_1, \mbf{T}_1, \mbf{T}_2, F, E)$,
i.e.,
\begin{eqnarray}
    \label{eq-security-of-dummy-key}
\frac{1}{2}\|
\rho_{S \mbf{W}_1 \mbf{T}_1 \mbf{T}_2 F E} -
\rho_{S}^{\rom{mix}} \otimes 
\rho_{\mbf{W}_1 \mbf{T}_1 \mbf{T}_2 F E}
\| \le \bar{\varepsilon}.
\end{eqnarray}
The assumption that Alice and Bob's bit sequence
are identical to $\mbf{u}$ with probability
$1 - \varepsilon_1$ implies that
$\hat{\mbf{w}}_1 = \mbf{w}_1$ and $\hat{\mbf{t}}_2 = \mbf{t}_2$
with probability $1- \varepsilon_1$. 
Since $(\mbf{u}, \hat{\mbf{u}})$, $(\mbf{w}_1, \hat{\mbf{w}}_1)$,
and $(\mbf{t}_2, \hat{\mbf{t}}_2)$ can be computed from 
$(\mbf{x}, \mbf{y})$, by using 
Lemma \ref{lemma-error-prob-continuity}, we have
\begin{eqnarray*}
\| 
\rho_{\mbf{X} \mbf{Y} \hat{\mbf{U}} \hat{\mbf{W}}_1 \mbf{T}_1
\hat{\mbf{T}}_2 F E} -
\rho_{\mbf{X} \mbf{Y} \mbf{U} \mbf{W}_1 \mbf{T}_1 \mbf{T}_2 F E} 
\|
\le 2 \varepsilon_1.
\end{eqnarray*}
Since the trace distance does not increase by CP maps, we have
\begin{eqnarray*}
\| 
\rho_{S_A \hat{\mbf{W}}_1 \mbf{T}_1 \hat{\mbf{T}}_2 F E} -
\rho_{S \mbf{W}_1 \mbf{T}_1 \mbf{T}_2 F E} 
\| \le 2 \varepsilon_1.
\end{eqnarray*}
Thus the statement that the dummy key $S$ is $\bar{\varepsilon}$-secret
 implies that the actual key $S_A$ is 
$(\bar{\varepsilon} + 2 \varepsilon_1)$-secret as follows:
\begin{eqnarray*}
\lefteqn{
\| 
\rho_{S_A \hat{\mbf{W}}_1 \mbf{T}_1 \hat{\mbf{T}}_2 F E} -
\rho_{S_A}^{\rom{mix}} \otimes 
\rho_{ \hat{\mbf{W}}_1 \mbf{T}_1 \hat{\mbf{T}}_2 F E}
\|
} \\
&\le& \|
\rho_{S_A \hat{\mbf{W}}_1 \mbf{T}_1 \hat{\mbf{T}}_2 F E} -
\rho_{S \mbf{W}_1 \mbf{T}_1 \mbf{T}_2 F E} 
\| \\
&& +  \|
\rho_{S \mbf{W}_1 \mbf{T}_1 \mbf{T}_2 F E} -
\rho_S^{\rom{mix}} \otimes 
\rho_{\mbf{W}_1 \mbf{T}_1 \mbf{T}_2 F E}
\| \\
&& + \|
\rho_S^{\rom{mix}} \otimes \rho_{\mbf{W}_1 \mbf{T}_1 \mbf{T}_2 F E} -
\rho_{S_A}^{\rom{mix}} \otimes 
\rho_{\hat{\mbf{W}}_1 \mbf{T}_1 \hat{\mbf{T}}_2 F E}
\|,
\end{eqnarray*}
where the first term is upper bounded by $2 \varepsilon_1$,
the second term is upper bounded by $\bar{\varepsilon}$,
and the third term is also upper bounded by
$2 \varepsilon_1$ because $\rho_S^{\rom{mix}} = \rho_{S_A}^{\rom{mix}}$.
The assumption of Eq.~(\ref{eq-probability-of-reconciled-key-agrees})
also implies that the distilled key is $\varepsilon_1$-correct.
Thus the distilled key pair $(S_A, S_B)$ is 
$(\bar{\varepsilon} + 3 \varepsilon_1)$-secure.

In order to prove Eq.~(\ref{eq-security-of-dummy-key}),
we use Lemma \ref{lemma-privacy-amplification},
which gives the relation between the security and
the length of the distilled key.
If the length $\ell$ of the distilled key by 
the privacy amplification satisfies
\begin{eqnarray}
    \label{eq-proof-of-first-argument}
\log(1/8 \varepsilon) + \ell
\le H_{\min}^{\sqrt{8 \varepsilon}}(
\rho_{\mbf{U} \mbf{W}_1 \mbf{T}_1 \mbf{T}_2 E} |
\mbf{W}_1 \mbf{T}_1 \mbf{T}_2 E
),
\end{eqnarray}
then the distilled key $S$ is $\bar{\varepsilon}$-secret.
By using Corollary \ref{lemma-weak-chain-rule}, we can lower
bound the r.h.s. of Eq.~(\ref{eq-proof-of-first-argument}) by
\begin{eqnarray*}
H_{\min}^{\varepsilon}(\rho_{\mbf{U} \mbf{W}_1 E} | \mbf{W}_1 E)
- nR - \overline{n}_0 R_0,
\end{eqnarray*}
because the size of messages $\mbf{T}_1$ and $\mbf{T}_2$ 
are upper bounded by $nR$ and $\overline{n}_0 R_0$ respectively.
Thus we have shown the statement of the theorem
for the first argument of the maximum
in Eq.~(\ref{eq-the-secure-key-rate}).

Since the syndrome $\mbf{T}_1$ is computed from
the sequence $\mbf{U}_1$, if the distilled key
$S$ is $\bar{\varepsilon}$-secret in the case that
Eve can access the sequence $\mbf{U}_1$, then the distilled key
$S$ is $\bar{\varepsilon}$-secret in the case
that Eve can only access the syndrome $\mbf{T}_1$
instead of the sequence $\mbf{U}_1$.
Again using Lemma \ref{lemma-privacy-amplification},
if the length of the distilled key satisfies
\begin{eqnarray}
      \label{eq-proof-of-second-argument}
\log(1/ 8 \varepsilon) + \ell \le
H_{\min}^{\sqrt{8 \varepsilon}}( \rho_{\mbf{U} \mbf{W}_1 \mbf{U}_1
\mbf{T}_2 E} |
\mbf{W}_1 \mbf{U}_1 \mbf{T}_2 E),
\end{eqnarray}
then the distilled key $S$ is $\bar{\varepsilon}$-secret.
Again using Corollary \ref{lemma-weak-chain-rule},
we can lower bound the r.h.s. of Eq.~(\ref{eq-proof-of-second-argument})
by
\begin{eqnarray*}
H_{\min}^{\varepsilon}(\rho_{\mbf{U} \mbf{W}_1 \mbf{U}_1 E} |
\mbf{W}_1 \mbf{U}_1 E) - \overline{n}_0 R_0.
\end{eqnarray*}
Thus we have shown the statement of the theorem for
the second argument of the maximum in
Eq.~(\ref{eq-the-secure-key-rate}).
\end{proof}

\subsection{Fluctuation of the actual error rate}
\label{subsec-fluctuation-of-error}

In this section, we show that the parameter estimation works
with high probability (Lemma \ref{lemma-parameter-estimation}).
Then, we show that the information reconciliation
protocol works for symmetric errors if the protocol universally works
for the i.i.d. errors that are close to the estimated error 
distributions in the parameter estimation protocol
(Lemma \ref{lemma-condition-on-ir}).
 
For the output $Q \in {\cal Q}$ of the parameter estimation
protocol, let 
\begin{eqnarray*}
\Gamma_\mu(Q) := \{
\sigma_{AB} \in {\cal P}({\cal H}_A \otimes {\cal H}_B) \mid
\| P_A^{\sigma_{AB}} - Q \| \le \mu
\}
\end{eqnarray*}
be a set of two-qubit density operators
that are compatible with the output $Q$ with a fluctuation $\mu$,
where $P_A^{\sigma_{AB}}$ denotes the probability distribution
of the outcomes when measuring $\sigma_{AB}$ by the POVM ${\cal M}$,
i.e., $P_A^{\sigma_{AB}}(a) := \rom{Tr} [M_a \sigma_{AB}]$.
When $\rho_m = \sigma_{AB}^{\otimes m}$ is a product state for 
$\sigma_{AB} \notin \Gamma_\mu(Q)$, then by the law of large numbers,
the probability such that the parameter estimation protocol
outputs the type $Q$ is negligible.
The following lemma generalize this statement to 
permutation-invariant states.
\begin{lemma}
      \label{lemma-parameter-estimation}
\cite[Lemma 6.2.2]{renner:05b}
Let $0 \le r \le \frac{1}{2} m$.
Moreover, let 
$\ket{\theta} \in {\cal H}_{ABE} := {\cal H}_A \otimes {\cal H}_B \otimes {\cal H}_E$,
and let $\rho_{A^m B^m E^m}^{\ket{\theta}}$ be a density operator on
$\rom{Sym}({\cal H}_{ABE}^{\otimes m}, \ket{\theta}^{\otimes m - r})$.
For any $\varepsilon_{\san{P}} > 0$,
if $\rom{Tr}_E \ket{\theta}\bra{\theta} \notin \Gamma_\mu(Q)$
for 
\begin{eqnarray}
\label{eq-definition-mu}
\mu = 2 \sqrt{ \frac{\log (1 / \varepsilon_{\san{P}})}{m} + h(r/m) + \frac{|{\cal W}|}{m}
\log (\frac{m}{2} + 1) },
\end{eqnarray}
then the probability such that the
parameter estimation protocol outputs $Q$ is at most $\varepsilon_{\san{P}}$, i.e.,
${\cal E}_Q(\rho_{A^m B^m}^{\ket{\theta}}) \le
 \varepsilon_{\san{P}}$.
\end{lemma}

For the POVM ${\cal M}_{XY}$, which is used for obtaining the raw
keys in the QKD protocol, let
$P_{XY}^{\sigma_{AB}}$ be the probability distribution
of the outcomes when measuring $\sigma_{AB}$ by the 
POVM ${\cal M}_{XY}$, i.e.,
$P_{XY}^{\sigma_{AB}}(x,y) := \rom{Tr}[ M_{xy} \sigma_{AB}]$.
For 
\begin{eqnarray}
\label{eq-definition-bar-mu}
\bar{\mu} = 2 \sqrt{ \frac{\log (1 / \varepsilon_2)}{2n} + h(r/2n) + 
\frac{\log (n + 1)}{n}  },
\end{eqnarray}
let 
\begin{eqnarray*}
{\cal Q}_{\bar{\mu}}(Q) := \{ P \in {\cal P}_{2n}(\mathbb{F}_2^2)
\mid \min_{\sigma_{AB} \in \Gamma_\mu(Q)}
\| P_{XY}^{\sigma_{AB}} - P \| \le \bar{\mu} \}
\end{eqnarray*}
be a subset of all types on $\mathbb{F}_2^2$.
Note that if we measure a product state $\sigma_{AB}^{\otimes 2n}$
of $\sigma_{AB} \in \Gamma_\mu(Q)$ by the product
POVM ${\cal M}_{XY}^{\otimes 2n}$, then the joint
type $P_{\mbf{x}\mbf{y}}$ of the outcomes is
contained in the set ${\cal Q}_{\bar{\mu}}(Q)$ with
high probability.

\begin{lemma}
    \label{lemma-condition-on-ir}
Let $\rho_{A^{2n}B^{2n}E^{2n}}^{\ket{\theta}}$ be a density operator
on $\rom{Sym}({\cal H}_{ABE}^{\otimes 2n}, \ket{\theta}^{\otimes
 2n-r})$.
Let $P_{\mbf{X}\mbf{Y}}^{\ket{\theta}} \in {\cal P}(\mathbb{F}_2^{2n}
 \times \mathbb{F}_2^{2n})$ be a probability distribution of the
 outcomes when measuring $\rho_{A^{2n} B^{2n}}^{\ket{\theta}}$ by
the POVM ${\cal M}_{XY}^{\otimes 2n}$. 
Assume that Alice and Bob's bit sequence after
the IR protocol are identical to $\mbf{u}$ with
probability at least $1- \varepsilon_1$ for any
probability distribution $P \in {\cal Q}_{\bar{\mu}}(Q)$, i.e.,
\begin{eqnarray}
P^{2n}(\{ (\mbf{x},\mbf{y}) : \hat{\mbf{u}} \neq \mbf{u}~\mbox{or}~ \tilde{\mbf{u}} \neq \mbf{u} \})
\le \varepsilon_1.
\end{eqnarray}
If $\rom{Tr}_E \ket{\theta}\bra{\theta} \in \Gamma_\mu(Q)$, then
we have
\begin{eqnarray}
     \label{eq-probability-of-reconciled-key-agrees-for-symmetric}
P_{\mbf{X}\mbf{Y}}^{\ket{\theta}}(
\{ (\mbf{x},\mbf{y}) : \hat{\mbf{u}} \neq \mbf{u}~\mbox{or}~
\tilde{\mbf{u}} \neq \mbf{u} \}) 
\le  L \varepsilon_1 + \varepsilon_2,
\end{eqnarray}
where $L := (2n+1)^3$, and
$\varepsilon_2$ is given in Eq.~(\ref{eq-definition-bar-mu}).
\end{lemma}

\begin{proof}
For each type $P \in {\cal P}_{2n}(\mathbb{F}_2 \times \mathbb{F})$,
let 
\begin{eqnarray*}
\gamma_P := \frac{|\{ (\mbf{x}, \mbf{y}) : 
\hat{\mbf{u}} \neq \mbf{u}~\mbox{or}~ \tilde{\mbf{u}} \neq \mbf{u} \}
\cap {\cal T}_P^{2n}|}{|{\cal T}_P^{2n}|}
\end{eqnarray*}
be the ratio of pairs of sequences in ${\cal T}_P^{2n}$ such that
Alice or Bob's sequences after the IR protocol are
not identical to $\mbf{u}$.
Since the distribution $P_{\mbf{X}\mbf{Y}}^{\ket{\theta}}$ is
permutation invariant,
we can rewrite the l.h.s. of 
Eq.~(\ref{eq-probability-of-reconciled-key-agrees-for-symmetric}) as
\begin{eqnarray}
      \label{eq-decomposition-into-good-bad}
\sum_{P \in {\cal Q}_{\bar{\mu}}(Q) } 
\gamma_P P_{\mbf{X}\mbf{Y}}^{\ket{\theta}}({\cal T}_P^{2n})
+ \sum_{P \notin {\cal Q}_{\bar{\mu}}(Q) }
\gamma_P P_{\mbf{X}\mbf{Y}}^{\ket{\theta}}({\cal T}_P^{2n}).
\end{eqnarray}
Since $\rom{Tr}_E \ket{\theta}\bra{\theta} \in \Gamma_\mu(Q)$,
by using Lemma \ref{lemma-statistic-of-symmetric-state}, 
the second term of Eq.~(\ref{eq-decomposition-into-good-bad})
is upper bounded by $\varepsilon_2$.

On the other hand, by using Eq.~(\ref{eq-inequality-of-type}),
we have
\begin{eqnarray*}
\varepsilon_1 \ge \gamma_P P^{2n}({\cal T}_P^{2n})
\ge \frac{\gamma_P}{(2n+1)^3}
\end{eqnarray*}
for any $P \in {\cal Q}_{\bar{\mu}}(Q)$.
Thus, the first term of Eq.~(\ref{eq-decomposition-into-good-bad})
is upper bounded by $(2n+1)^3 \varepsilon_1$.
\end{proof}

\subsection{Security poof}

In order to save space, we abbreviate $2n+m$ by $K$.
In this section,
if there are two  operators $\rho \in {\cal P}({\cal H})$ and 
$\tilde{\rho} \in {\cal P}({\cal H})$, 
then the former represents the normalized density operator of the latter,
i.e., $\rho = \frac{1}{\rom{Tr}\tilde{\rho}} \tilde{\rho}$. 

\subsubsection{Parameter estimation}

We first analyze the situation after the parameter estimation
protocol is executed. More specifically, 
by using Lemmas \ref{de-finneti-theorem} and
\ref{lemma-parameter-estimation}, 
we will
show Eq.~(\ref{eq-proof-2}), which states that the density operator
$\rho_{A^{2n}B^{2n}E^{2n}}^Q$ after the parameter
estimation protocol can be approximated by a convex combination
of almost product states.

Since the tripartite state $\rho_{A^N B^N E^N}$ lies on the symmetric
subspace of ${\cal H}_{ABE}^{\otimes N} := ({\cal H}_A \otimes {\cal H}_B
\otimes {\cal H})^{\otimes N}$, by using Lemma \ref{de-finneti-theorem},
the density operator $\rho_{A^K B^K E^K}$ is
approximated by a convex combination of almost product states, i.e.,
\begin{eqnarray*}
\| 
\rho_{A^K B^K E^K} -
\int_{{\cal S}_1} \rho_{A^K B^K E^K}^{\ket{\theta}} \nu(\ket{\theta})
\| \le \kappa,
\end{eqnarray*}
where the integral runs over the set ${\cal S}_1 := {\cal S}_1({\cal
H}_{ABE})$ of normalized vectors on ${\cal H}_{ABE}$, where 
\begin{eqnarray*}
\rho_{A^K B^K E^K}^{\ket{\theta}} \in 
{\cal P}(\rom{Sym}({\cal H}_{ABE}^{\otimes K}, \ket{\theta}^{\otimes K -r}))
\end{eqnarray*}
for any $\ket{\theta} \in {\cal S}_1$, and where 
\begin{eqnarray*}
r := \frac{2N}{k} \{
\ln(2/\kappa) + \dim({\cal H}_A \otimes {\cal H}_B) \cdot
\ln k
\}.
\end{eqnarray*}
Since the trace distance does not increase by applying a CP map,
we have
\begin{eqnarray}
     \label{eq-proof-1}
\|
\tilde{\rho}_{A^{2n}B^{2n}E^{2n}}^Q -
\int_{{\cal S}_1} \tilde{\rho}_{A^{2n} B^{2n} E^{2n}}^{Q, \ket{\theta}}
\nu(\ket{\theta})
\| \le \kappa,
\end{eqnarray}
where 
\begin{eqnarray*}
\tilde{\rho}_{A^{2n} B^{2n} E^{2n}}^{Q, \ket{\theta}} :=
(\rom{id}_{A^{2n} B^{2n}} \otimes {\cal E}_Q
\otimes \rom{id}_{E^{2n}})(\rho_{A^K B^K E^{2n}}^{\ket{\theta}}).
\end{eqnarray*}
Let 
\begin{eqnarray*}
{\cal V}_\mu := \{
\ket{\theta} \in {\cal S}_1 \mid
\rom{Tr}_E \ket{\theta}\bra{\theta} \in \Gamma_\mu
\}
\end{eqnarray*}
be the subset of ${\cal S}_1$ that is compatible with 
the output $Q$ of the parameter estimation protocol
with the fluctuation $\mu$.
From Lemma \ref{lemma-parameter-estimation}, if $\ket{\theta} \notin
\Gamma_\mu(Q)$,
then the probability such that the parameter estimation protocol
outputs $Q$ is at most $\varepsilon_{\san{P}}$, i.e.,
$\| \tilde{\rho}_{A^{2n} B^{2n} E^{2n}}^{Q, \ket{\theta}} \| \le
\varepsilon_{\san{P}}$.
Thus, we can restrict the integral in Eq.~(\ref{eq-proof-1}) to the
set ${\cal V}_\mu$ as 
\begin{eqnarray*}
\lefteqn{ \|
\tilde{\rho}_{A^{2n} B^{2n} E^{2n}}^Q -
\tilde{\rho}_{A^{2n} B^{2n} E^{2n}}^{Q, {\cal V}_\mu}
\| } \\
&\le& \|
\tilde{\rho}_{A^{2n} B^{2n} E^{2n}}^Q -
\int_{{\cal S}_1} 
\tilde{\rho}_{A^{2n} B^{2n} E^{2n}}^{Q, \ket{\theta}} \nu(\ket{\theta}) 
\| \\
&& +
\|
\int_{{\cal V}_\mu^{\rom{c}}} 
\tilde{\rho}_{A^{2n} B^{2n} E^{2n}}^{Q, \ket{\theta}} \nu(\ket{\theta}) 
\| 
\le \kappa + \varepsilon_{\san{P}},
\end{eqnarray*}
where we set
\begin{eqnarray*}
\tilde{\rho}_{A^{2n} B^{2n} E^{2n}}^{Q, {\cal V}_\mu} :=
\int_{{\cal V}_\mu} 
\tilde{\rho}_{A^{2n} B^{2n} E^{2n}}^{Q, \ket{\theta}} \nu(\ket{\theta}),
\end{eqnarray*}
and ${\cal V}_\mu^{\rom{c}}$ is the complement of ${\cal V}_\mu$
in ${\cal S}_1$.
By using the following Lemma \ref{lemma-normalizing},
the normalized version of the operators satisfy
\begin{eqnarray}
     \label{eq-proof-2}
\|
\rho_{A^{2n} B^{2n} E^{2n}}^Q -
\rho_{A^{2n} B^{2n} E^{2n}}^{Q, {\cal V}_\mu}
\| \le 2 \tilde{\tau},
\end{eqnarray}
where $\tilde{\tau} := \frac{\kappa +
\varepsilon_{\san{P}}}{P_{\san{PE}}(Q)}$.

\begin{lemma}
     \label{lemma-normalizing}
Let $\tilde{\rho}, \tilde{\sigma} \in {\cal P}({\cal H})$
be (not necessarily normalized) operators.
Assume that $\| \tilde{\rho} - \tilde{\sigma} \| \le \varepsilon$
for $\varepsilon \ge 0$. Let $\rho :=
 \frac{1}{\rom{Tr}\tilde{\rho}}\tilde{\rho}$
and $\sigma := \frac{1}{\rom{Tr}\tilde{\sigma}} \sigma$ be the 
normalized operators. Then, we have 
$\| \rho - \sigma \| \le 2 \tilde{\varepsilon}$ 
for $\tilde{\varepsilon} := \frac{\varepsilon}{\rom{Tr} \tilde{\rho}}$.
\end{lemma}
\begin{proof}
From the assumption, we have $\| \rho - \hat{\sigma} \| \le
 \tilde{\varepsilon}$,
where $\hat{\sigma} := \frac{1}{\rom{Tr} \tilde{\rho}} \sigma$.
By using the triangle inequality, we have
\begin{eqnarray*}
1- \tilde{\varepsilon} \le
\| \rho \| - \| \rho - \hat{\sigma} \|
\le \| \hat{\sigma} \| 
\le \| \rho \| + \| \hat{\sigma} - \rho \| 
\le 1 + \tilde{\varepsilon}. 
\end{eqnarray*}
Thus, we have 
\begin{eqnarray*}
\| \sigma - \hat{\sigma} \| = | 1 - \| \hat{\sigma} \| | \le 
\tilde{\varepsilon}.
\end{eqnarray*}
Using once again the triangle inequality, we have
\begin{eqnarray*}
\| \rho - \sigma \| \le
\| \rho - \hat{\sigma} \| + \| \hat{\sigma} - \sigma \|
\le 2 \tilde{\varepsilon}.
\end{eqnarray*}
\end{proof}

\subsubsection{Information reconciliation}

According to  Section 2,
the IR protocol universally works with a negligible error 
probability for i.i.d. errors, if we set the parameters
$R(Q) = H(P_{W_1}) + \delta$, 
$R_0(Q) = H(P_{W_2|W_1=0}) + \delta$, 
$\frac{\overline{n}_0}{n} = P_{W_1}(0) + \delta$, and
$\frac{\underline{n}_0}{n} = P_{W_1}(0) + \delta$.
In this section, by using Lemma \ref{lemma-condition-on-ir},
we show that 
the IR protocol also works 
with a negligible error probability in the QKD protocol, i.e.,
\begin{eqnarray}
     \label{eq-proof-3}
P_{\mbf{X}\mbf{Y}}^Q( \{ (\mbf{x}, \mbf{y}) :
\hat{\mbf{u}} \neq \mbf{u}~\mbox{or}~\tilde{\mbf{u}} \neq \mbf{u} \})
\le \varepsilon_1 + \varepsilon_2 + 2 \tilde{\tau},
\end{eqnarray}
where $P_{\mbf{X}\mbf{Y}}^Q$ is the probability distribution of 
the outcomes when measuring $\rho_{A^{2n} B^{2n}}^Q$ by 
${\cal M}_{XY}^{\otimes 2n}$.
Note that $\varepsilon_1$ is the error probability of
the IR protocol for i.i.d. errors, which exponentially
goes to $0$ as $n \to \infty$ if we use appropriate
linear codes \cite[Corollary 2]{csiszar:82}. As we will see later,
$\varepsilon_2$ also exponentially goes to $0$ as $n \to \infty$.

By using the fact that the trace distance does not increase
by the CP map (measurement by ${\cal M}_{XY}^{\otimes 2n}$),
the l.h.s. of Eq.~(\ref{eq-proof-3}) is upper bounded by
\begin{eqnarray*}
P_{\mbf{X}\mbf{Y}}^{Q, {\cal V}_\mu}(\{ (\mbf{x}, \mbf{y}) :
\hat{\mbf{u}} \neq \mbf{u}~\mbox{or}~\tilde{\mbf{u}} \neq \mbf{u} \})
+ 2 \tilde{\tau},
\end{eqnarray*}
where $P_{\mbf{X}\mbf{Y}}^{Q, {\cal V}_\mu}$ is the probability
distribution of the outcomes when measuring $\rho_{A^{2n} B^{2n}}^{Q,
{\cal V}_\mu}$ by ${\cal M}_{XY}^{\otimes 2n}$.
Since $\rho_{A^{2n} B^{2n} E^{2n}}^{Q, {\cal V}_\mu}$ is a convex
combination of density operators
$\rho_{A^{2n} B^{2n} E^{2n}}^{Q, \ket{\theta}}$ on
$\rom{Sym}({\cal H}_{ABE}^{\otimes 2n}, \ket{\theta}^{\otimes 2n-r})$
such that $\rom{Tr}_E[ \ket{\theta}\bra{\theta}] \in \Gamma_\mu(Q)$,
by using Lemma \ref{lemma-condition-on-ir}, we have Eq.~(\ref{eq-proof-3}).

\subsubsection{Privacy amplification}

In this section, we analyze the PA protocol.
By applying Theorem \ref{theorem-key-distillation-from-raw-key},
if the length $\ell(Q)$ of the distilled key satisfies
\begin{eqnarray}
\lefteqn{ \ell(Q) \le } \nonumber \\
&& \max[ 
H_{\min}^{\tilde{\varepsilon}}(\rho_{\mbf{U} \mbf{W}_1 E^N}^Q |
\mbf{W}_1 E^N) 
- n R(Q) - \overline{n}_0 R_0(Q), \nonumber \\
&& ~~~~~~H_{\min}^{\tilde{\varepsilon}}(\rho_{\mbf{U} \mbf{W}_1 \mbf{U}_1 E^N}
 | \mbf{W}_1 \mbf{U}_1 E^N)  
 - \overline{n}_0(Q) R_0(Q)
] \nonumber \\
&&  ~~~~~~~~- \log(1/ 8 \tilde{\varepsilon}),
    \label{eq-proof-4}
\end{eqnarray}
then the distilled key is 
$(3 \sqrt{2 \tilde{\varepsilon}} + 3 \varepsilon_1 + 3
\varepsilon_2 + 6 \tilde{\tau})$-secure, where 
$\rho_{\mbf{U} \mbf{W}_1 E^N}^Q$ and
$\rho_{\mbf{U} \mbf{W}_1 \mbf{U}_1 E^N}$ are
derived from $\rho_{\mbf{X}\mbf{Y} E^N}^Q$ 
by using functions $\xi_1$ and $\xi_2$
in the same way
as in Section 2.
Let $\tilde{\varepsilon} := \sqrt{24 \tilde{\tau}}$.
Multiplying the probability $P_{\san{PE}}(Q)$, 
the quantities 
$P_{\san{PE}}(Q) 3 \sqrt{2 \tilde{\varepsilon}} \le 6 (6 (\kappa +
\varepsilon_{\san{P}}))^{1/4}$
and $P_{\san{PE}}(Q) 6 \tilde{\tau} = 6 (\kappa +
\varepsilon_{\san{P}})$
goes to $0$ as $\kappa, \varepsilon_{\san{P}} \to 0$.
Thus, the security of the distilled key, i.e.,
the l.h.s. of Eq.~(2) goes
to $0$ as $\kappa, \varepsilon_{\san{P}}, \varepsilon_1, \varepsilon_2
\to 0$.

\subsubsection{Evaluation of key rate}

One more thing we have left is to replace the r.h.s. of
Eq.~(\ref{eq-proof-4}) by smaller but  more concise equation.
Noting that $\kappa + \varepsilon_{\san{P}} \le \tilde{\varepsilon}$,
we can replace the last term $\log(1/ 8 \tilde{\varepsilon})$ by
$\log(1/8(\kappa + \varepsilon_{\san{P}}))$.

Let $\rho_{\mbf{X}\mbf{Y} E^{2n}}^{Q, \ket{\theta}} := ({\cal
E}_{XY}^{\otimes 2n} \otimes \rom{id}_{E^{2n}})(\rho_{A^{2n} B^{2n}
E^{2n}}^{Q, \ket{\theta}})$, and let 
$\rho_{\mbf{U} \mbf{W}_1 E^{2n}}^{Q, \ket{\theta}}$ be the density
operator derived from $\rho_{\mbf{X}\mbf{Y} E^{2n}}^{Q, \ket{\theta}}$ 
in the same way as in 
Section 2.
Since $\rho_{A^{2n} B^{2n} E^{2n}}^{Q, \ket{\theta}}$
lies on $\rom{Sym}({\cal H}_{A^2 B^2 E^2}^{\otimes n},
\ket{\theta^2}^{\otimes n-r})$ for $\ket{\theta^2} :=
\ket{\theta}^{\otimes 2}$, we can use 
Lemma \ref{lemma-min-entropy-of-symmetric-state}
to obtain
\begin{eqnarray}
\lefteqn{ \frac{1}{n} H_{\min}^{(\kappa + \varepsilon_{\san{P}})}(
\rho_{\mbf{U} \mbf{W}_1 E^{2n}}^{Q, \ket{\theta}} | \mbf{W}_1 E^{2n}) }
\nonumber  \\
&&~~\ge H_{\sigma}(U_1 U_2 | W_1 E_1 E_2) - \delta^\prime,
     \label{eq-proof-5}
\end{eqnarray}
where 
\begin{eqnarray*}
\delta^\prime := 9 \sqrt{
\frac{2 \log(4/ (\kappa + \varepsilon_{\san{P}}))}{n} +
h(r/n)
},
\end{eqnarray*}
and where $\sigma_{U_1 U_2 W_1 E_1 E_2}$ is derived 
from $\sigma_{X_1 X_2 Y_1 Y_2 E_1 E_2} :=
({\cal E}_{XY}^{\otimes 2} \otimes \rom{id}_E^{\otimes 2})(
\ket{\theta}\bra{\theta}^{\otimes 2})$ in the same way as in
Section 2.

Let $\rho_{\mbf{U} \mbf{W}_1 E^{2n}}^{Q, {\cal V}_\mu}$ be a density
operator derived from 
$\rho_{\mbf{X} \mbf{Y} E^{2n}}^{Q, {\cal V}_\mu} :=
({\cal E}_{XY}^{\otimes 2} \otimes \rom{id}_{E^{2n}})(\rho_{A^{2n}
B^{2n} E^{2n}}^{Q, {\cal V}_\mu})$ in the same way as in 
Section 2.
Since $\rho_{\mbf{U} \mbf{W}_1 E^{2n}}^{Q, {\cal V}_\mu}$ is a
convex combination of density operators 
$\rho_{\mbf{U} \mbf{W}_1 E^{2n}}^{Q, \ket{\theta}}$, by using
Eqs.~(\ref{strong-subadditivity}) 
and (\ref{conditioning-on-classical-information-2})
in Lemma \ref{lemma-properties-of-min-max-entropy}, we have
\begin{eqnarray}
\lefteqn{ H_{\min}^{(\kappa + \varepsilon_{\san{P}})}(
\rho_{\mbf{U} \mbf{W}_1 E^{2n}}^{Q, {\cal V}_\mu} |
\mbf{W}_1 E^{2n} )  } \nonumber \\
&&~~\ge \min_{\ket{\theta} \in {\cal V}_\mu}
H_{\min}^{(\kappa + \varepsilon_{\san{P}})}(
\rho_{\mbf{U} \mbf{W}_1 E^{2n}}^{Q, \ket{\theta}} |
\mbf{W}_1 E^{2n} ).
    \label{eq-proof-6}
\end{eqnarray}

Since the trace distance does not increase by a CP map, we have
\begin{eqnarray}
\| 
\rho_{\mbf{U} \mbf{W}_1 E^{2n}}^Q - 
\rho_{\mbf{U} \mbf{W}_1 E^{2n}}^{Q, {\cal V}_\mu}
\| \le 2 \tilde{\tau}.
     \label{eq-proof-7}
\end{eqnarray}
By using (a) Lemmas \ref{lemma-chain-rule} and
\ref{lemma-chain-2},
(b) Eq.~(\ref{eq-proof-7}) and Lemma \ref{lemma-continuity},
(c) $\kappa + \varepsilon_{\san{P}} \le \tilde{\tau}$,
(d) Eqs.~(\ref{eq-proof-5}) and (\ref{eq-proof-6}),
we have 
\begin{widetext}
\begin{eqnarray*}
\frac{1}{n} H_{\min}^{\sqrt{24 \tilde{\tau}}}(\rho_{\mbf{U} \mbf{W}_1
E^N}^Q| \mbf{W}_1 E^N ) 
&\stackrel{\mbox{{\tiny (a)}}}{\ge}& 
\frac{1}{n}
H_{\min}^{3 \tilde{\tau}}(\rho_{\mbf{U} \mbf{W}_1 E^{2n}}^Q | \mbf{W}_1
E^{2n}) - \frac{2 (m+k)}{n} \log \dim {\cal H}_E \\
&\stackrel{\mbox{{\tiny (b)}}}{\ge}&
\frac{1}{n} H_{\min}^{\tilde{\tau}}(\rho_{\mbf{U} \mbf{W}_1
E^{2n}}^{Q,{\cal V}_\mu} | \mbf{W}_1
E^{2n}) - \frac{2 (m+k)}{n} \log \dim {\cal H}_E \\
&\stackrel{\mbox{{\tiny (c)}}}{\ge}&
\frac{1}{n} H_{\min}^{(\kappa + \varepsilon_{\san{P}})}(
\rho_{\mbf{U} \mbf{W}_1 E^{2n}}^{Q, {\cal V}_\mu} | \mbf{W}_1
E^{2n}) - \frac{2 (m+k)}{n} \log \dim {\cal H}_E \\
&\stackrel{\mbox{{\tiny (d)}}}{\ge}&
\min_{\ket{\theta} \in {\cal V}_\mu}
H_{\sigma}(U_1 U_2 | W_1 E_1 E_2) - \delta^\prime
- \frac{2 (m+k)}{n} \log \dim {\cal H}_E .
\end{eqnarray*}
In a similar manner, we have
\begin{eqnarray*}
\frac{1}{n} H_{\min}^{\sqrt{24 \tilde{\tau}}}(\rho_{\mbf{U} \mbf{W}_1 \mbf{U}_1
E^N}^Q| \mbf{W}_1 \mbf{U}_1 E^N ) 
\ge 
\min_{\ket{\theta} \in {\cal V}_\mu} 
H_{\sigma}(U_2 | W_1 U_1 E_1 E_2) - \delta^\prime
- \frac{2 (m+k)}{n} \log \dim {\cal H}_E .
\end{eqnarray*}
\end{widetext}
Finally, setting $k := \alpha_1 n$,
$m := \alpha_2 n$, $\kappa := e^{- \alpha_3 k}$,
$\varepsilon_{\san{P}} := 2^{- \alpha_4 m}$,
$\varepsilon_2 := 2^{- \alpha_5 n}$,
and taking $n \to \infty$
and $\alpha_1,\alpha_2,\alpha_3, \alpha_4, \alpha_5 \to 0$,
we have the assertion of theorem.

\section{Proof of Theorem 3}
\label{proof-of-theorem-3}

This section presents a proof of Theorem 3 in the main text.

Let 
\begin{eqnarray*}
\ket{\psi_{ABE}} &:=& \sum_{\san{x}, \san{z} \in \mathbb{F}_2 }
\sqrt{P_{\san{XZ}}(\san{x},\san{z})} \ket{\psi(\san{x},\san{z})} 
\ket{\san{x}, \san{z}} \\
&=& \sum_{x, \san{x} \in \mathbb{F}_2}
\sqrt{P_{\san{X}}(\san{x})} \ket{x, x + \san{x}} \ket{\phi(x,\san{x})}
\end{eqnarray*}
be a purification of $\sigma_{AB} = \sum_{\san{x},\san{z} \in
\mathbb{F}_2} \ket{\psi(\san{x},\san{z})} \bra{\psi(\san{x},\san{z})}$,
where we set
\begin{eqnarray*}
\ket{\phi(x, \san{x})} := \frac{1}{\sqrt{P_{\san{X}}}(\san{x})}
\sum_{\san{z} \in \mathbb{F}_2} (-1)^{x \san{z}} 
\sqrt{P_{\san{XZ}}(\san{x}, \san{z})} \ket{\san{x}, \san{z}},
\end{eqnarray*}
and where $P_{\san{X}}(\san{x}) = \sum_{\san{z} \in \mathbb{F}_2}
P_{\san{XZ}}(\san{x},\san{z})$
is a marginal distribution.
Then, let
\begin{eqnarray*}
\lefteqn{ \sigma_{X_1 X_2 Y_1 Y_2 E_1 E_2} } \\
&:=&
({\cal E}_{XY}^{\otimes 2} \otimes \rom{id}_E^{\otimes 2})
(\ket{\psi_{ABE}}\bra{\psi_{ABE}}^{\otimes 2}) \\
&=& \sum_{\vec{x}, \vec{\san{x}} \in \mathbb{F}_2^2} \frac{1}{4}
P_{\san{X}}^2(\vec{\san{x}}) \ket{\vec{x}, \vec{x} + \vec{\san{x}}}
\bra{\vec{x}, \vec{x} + \vec{\san{x}}} \otimes \sigma_{E_1
E_2}^{\vec{x}, \vec{\san{x}}},
\end{eqnarray*}
where 
\begin{eqnarray*}
\sigma_{E_1 E_2}^{\vec{x}, \vec{\san{x}}} :=
\ket{\phi(x_1, \san{x}_1)} \bra{\phi(x_1, \san{x}_1)}
\otimes \ket{\phi(x_2, \san{x}_2)} \bra{\phi(x_2, \san{x}_2)}
\end{eqnarray*}
for $\vec{x} = (x_1,x_2)$ and $\vec{\san{x}} = (\san{x}_1,\san{x}_2)$.

Noting that 
\begin{eqnarray*}
P_{X_1 X_2 Y_1 Y_2}(\vec{x}, \vec{x} + \vec{\san{x}}) =
\frac{1}{4} P_{\san{X}}^2(\vec{\san{x}}),
\end{eqnarray*}
we have
\begin{eqnarray*}
P_{U_1}(u_1) &=& \frac{1}{2} \\
P_{W_1}(w_1) &=& \sum_{\vec{x} \in \mathbb{F}_2^2
 \atop \san{x}_1 + \san{x}_2 = w_1 } 
P_{\san{X}}^2(\vec{\san{x}}) \\
P_{U_2|W_1 = 0}(u_2) &=& \frac{1}{2} \\
P_{U_2|W_1 = 1}(u_2) &=& 1 \\
P_{W_2|W_1=0}(w_2) &=& \frac{P_{\san{X}}^2(w_2,w_2)}{P_{W_1}(w_1)} \\
P_{W_2|W_1=1}(0) &=& 1.
\end{eqnarray*}
Using these formulas, we can write
\begin{eqnarray*}
&& \sigma_{U_1 U_2 W_1 E_1 E_2} =
\sum_{\vec{u} \in \mathbb{F}_2^2} \sum_{w_1 \in \mathbb{F}_2}
P_{U_1}(u_1) P_{W_1}(w_1) \\
&& ~~~~~~~~~P_{U_2|W_1 = w_1}(u_2)  
 \ket{\vec{u}, w_1}\bra{\vec{u}, w_1} \otimes
\bar{\sigma}_{E_1 E_2}^{\vec{u}, w_1} 
\end{eqnarray*}
for $\vec{u} = (u_1,u_2)$, where 
\begin{eqnarray*}
\bar{\sigma}_{E_1 E_2}^{\vec{u},w_1} :=
\sum_{w_2 \in \mathbb{F}_2} P_{W_2|W_1=0}(w_2)
\sigma_{E_1 E_2}^{\vec{u} G, (w_1,w_2)G}
\end{eqnarray*}
for $w_1 = 0$ and a matrix $G = \left(\begin{array}{cc} 1 & 1 \\ 1 & 0
				      \end{array} \right)$,
and 
\begin{eqnarray*}
\bar{\sigma}_{E_1 E_2}^{\vec{u}, w_1} :=
\sum_{a,b \in \mathbb{F}_2} \frac{1}{4} \sigma_{E_1 E_2}^{(u_1,a)G, (w_1,b)G}
\end{eqnarray*}
for $w_1 = 1$.

Since supports of rank $1$ matrices $\{ \sigma_{E_1 E_2}^{\vec{x},
\vec{\san{x}}} \}_{\vec{\san{x}} \in \mathbb{F}_2^2}$ are
orthogonal to each other, 
$\sigma_{E_1 E_2}^{\vec{u}, w_1}$ for $w_1 = 0$ is already
eigen value decomposed.
Applying Lemma \ref{lemma-for-key-rate}
for $\san{J} = \{ 00, 10\}$ and 
$C = C^{\bot} = \{ 00, 11 \}$, we can 
eigen value decompose $\sigma_{E_1 E_2}^{\vec{u}, w_1}$
for $w_1 = 1$ as 
\begin{eqnarray*}
\sigma_{E_1 E_2}^{\vec{u}, w_1} =
\sum_{b \in \mathbb{F}_2} \frac{1}{2}
\sum_{\vec{\san{j}} \in \san{J} } 
P_{\san{J}|\vec{\san{X}}=\vec{\san{x}}}(\vec{\san{j}})
\ket{\vartheta((u_1,0), \san{x}, \vec{\san{j}})}
\bra{\vartheta((u_1,0), \san{x}, \vec{\san{j}})},
\end{eqnarray*}
where we follow the notations in 
Lemma \ref{lemma-for-key-rate}
for $m=2$.

Thus, we have
\begin{eqnarray}
\lefteqn{ H(\sigma_{U_1 U_2 W_1 E_1 E_2}) } \nonumber \\
&=& H(P_{U_1}) + H(P_{W_1}) + \sum_{w_1 \in \mathbb{F}_2}
P_{W_1}(w_1) \{ H(P_{U_2|W_1 = w_1}) \nonumber \\
&& +
\sum_{\vec{u} \in \mathbb{F}_2^2} P_{U_1}(u_1) P_{U_2|W_1 = w_1}(u_2)
H(\sigma_{E_1 E_2}^{\vec{u}, w_1}) \} \nonumber \\
&=& 1 + H(P_{\bar{\san{X}}}) + P_{\bar{\san{X}}}(0)\{
1 + H(P_{\vec{\san{X}}|\bar{\san{X}}=0})\} \nonumber \\ 
&&~~ +
P_{\bar{\san{X}}}(1) H(P_{\vec{\san{X}}\san{J}| \bar{\san{X}}=1}). 
\label{entropy-uwe}
\end{eqnarray}

Taking the partial trace of $\sigma_{U_1 U_2 W_1 E_1 E_2}$
over systems $U_1, U_2$, we have
\begin{eqnarray*}
\sigma_{W_1 E_1 E_2} &=& \sum_{w_1 \in \mathbb{F}_2}
P_{W_1}(w_1) \ket{w_1}\bra{w_1}  \\
&& \otimes
\left(
\sum_{\vec{u} \in \mathbb{F}_2^2} P_{U_1}P_{U_2|W_1 = w_1}(u_2)
\bar{\sigma}_{E_1 E_2}^{\vec{u},w_1}
\right).
\end{eqnarray*}
Thus, we have
\begin{eqnarray}
H(\sigma_{W_1 E_1 E_2}) &=&
H(P_{W_1}) + \sum_{w_1 \in \mathbb{F}_2} P_{W_1}(w_1) \nonumber \\
&& H \left(
\sum_{\vec{u} \in \mathbb{F}_2^2} P_{U_1}P_{U_2|W_1 = w_1}(u_2)
\bar{\sigma}_{E_1 E_2}^{\vec{u},w_1}
\right) \nonumber \\
&=& 
H(P_{\bar{\san{X}}}) + 
\sum_{\bar{\san{x}} \in \mathbb{F}_2}
P_{\bar{\san{X}}}(0) H(P_{\vec{\san{X}}\vec{\san{Z}}| \bar{\san{X}}= \bar{\san{x}}}). 
\label{entropy-we}
\end{eqnarray}

Combining Eqs.~(\ref{entropy-uwe}) and (\ref{entropy-we}), we have
\begin{eqnarray*}
\lefteqn{ H_{\sigma}(U_1 U_2 | W_1 E_1 E_2) - H(P_{W_1}) 
P_{W_1}(0) H(P_{W_2|W_1 = 0}) } \\
&=&
2 - H(P_{\vec{\san{X}}\vec{\san{Z}}}) +
P_{\bar{\san{X}}}(1) \{
H(P_{\vec{\san{X}}\san{J}| \bar{\san{X}}=1}) - 1
\} \\
&=& 
2 - 2 H(P_{\san{X}\san{Z}}) +
P_{\bar{\san{X}}}(1) 
h\left( \frac{p_{00} p_{10} + p_{01} p_{11}}{
(p_{00} + p_{01})(p_{10} + p_{11})}\right).
\end{eqnarray*}

On the other hand, by taking partial trace of
$\sigma_{U_1 U_2 W_1 E_1 E_2}$ over the system
$U_1$, we have
\begin{eqnarray*}
\sigma_{U_1 W_1 E_1 E_2} &=&
\sum_{u_1, w_1 \in \mathbb{F}_2} \frac{1}{2} P_{W_1}(w_1) 
\ket{u_1,w_1}\bra{u_1,w_1} \\
&& \otimes
\left(
\sum_{u_2 \in \mathbb{F}_2} P_{U_2 | W_1 = w_1}(u_2)
\sigma_{E_1 E_2}^{(u_1,u_2), w_1}
\right) .
\end{eqnarray*}
Thus, we have
\begin{eqnarray}
H(\sigma_{U_1 W_1 E_1 E_2}) &=&
1 + H(P_{W_1}) + \sum_{u_1, w_1 \in \mathbb{F}_2}
\frac{1}{2} P_{W_1}(w_1) \nonumber \\
&& ~~
H\left(
\sum_{u_2 \in \mathbb{F}_2} P_{U_2 | W_1 = w_1}(u_2)
\sigma_{E_1 E_2}^{(u_1,u_2), w_1}
\right) \nonumber \\
&=& 1 + H(P_{\bar{\san{X}}}) + 
\sum_{\bar{\san{x}} \in \mathbb{F}_2} P_{\bar{\san{X}}}(\bar{\san{x}})
H(P_{\vec{\san{X}} \san{J} | \bar{\san{X}}=1}). \nonumber \\
\label{entropy-uwe2}
\end{eqnarray}

Combining Eqs.~(\ref{entropy-uwe}) and (\ref{entropy-uwe2}), we have
\begin{eqnarray*}
&&  H_{\sigma}(U_2|W_1 U_1 E_1 E_2) - P_{W_1}(0) 
H(P_{W_2|W_1=0})  \\
&=& P_{\bar{\san{X}}}(0) (
1- H(P_{\san{X}\san{Z}}^\prime)
). 
\end{eqnarray*}

\begin{lemma}
\label{lemma-for-key-rate}
Let $C$ be a linear subspace of $\mathbb{F}_2^m$.
Let 
\begin{eqnarray*}
\ket{\varphi^m(\vec{x},\vec{\san{x}})} 
\frac{1}{\sqrt{P_{\san{X}}^m(\vec{\san{x}})}}
\sum_{\vec{\san{z}} \in \mathbb{F}_2^m}
(-1)^{\vec{x} \cdot \vec{\san{z}}}
\sqrt{P_{\san{XZ}}^m(\vec{\san{x}},\vec{\san{z}})}
 \ket{\vec{\san{x}},\vec{\san{z}}}, 
\end{eqnarray*}
and 
$\sigma_{E^m}^{\vec{x},\vec{\san{x}}} := \ket{\varphi^m(\vec{x},\vec{\san{x}})} 
\bra{\varphi^m(\vec{x},\vec{\san{x}})}$.
Let $\san{J}$ be a set of coset representatives of the cosets $\mathbb{F}_2^m/C$,
and
\begin{eqnarray*}
P_{\san{J}|\san{X}^m = \vec{\san{x}}}(\vec{\san{j}}) :=
\frac{ \sum_{\vec{\san{c}} \in C^\bot} P_{\san{XZ}}^m(\vec{\san{x}}, \vec{\san{j}}+\vec{\san{c}})}
{P_{\san{X}}^m(\vec{\san{x}})}
\end{eqnarray*}
be  conditional probability distributions on $\san{J}$.
Then, for any $\vec{a} \in \mathbb{F}_2^m$, we have
\begin{eqnarray}
\label{eq-mixture-of-eve-state}
\sum_{\vec{x} \in C} \frac{1}{|C|} \sigma_{E^m}^{\vec{x} + \vec{a}, \vec{\san{x}}} =
\sum_{\vec{\san{j}} \in \san{J}} P_{\san{J}|\san{X}^m=\vec{\san{x}}}(\vec{\san{j}})
\ket{\vartheta(\vec{a},\vec{\san{x}}, \vec{\san{j}})}
\bra{\vartheta(\vec{a},\vec{\san{x}},\vec{\san{j}})},
\end{eqnarray}
where 
\begin{eqnarray*}
\ket{\vartheta(\vec{a},\vec{\san{x}},\vec{\san{j}})} &:=& 
\frac{1}{\sqrt{\sum_{\vec{\san{e}} \in C^\bot} 
P_{\san{XZ}}^m(\vec{\san{x}}, \vec{\san{j}}+\vec{\san{e}})} } \\
&& \sum_{\vec{\san{c}} \in C^\bot} (-1)^{\vec{a} \cdot \vec{\san{c}}} 
 \sqrt{P_{\san{XZ}}^m(\vec{\san{x}}, \vec{\san{j}}+\vec{\san{c}})}
\ket{\vec{\san{x}}, \vec{\san{j}} + \vec{\san{c}} }.
\end{eqnarray*}
\end{lemma}
\begin{remark}
\label{remark-of-lemma-for-key-rate}
If $\vec{\san{j}} \neq \vec{\san{i}}$, obviously we have 
$\langle \vartheta(\vec{a},\vec{\san{x}},\vec{\san{j}}) |
\vartheta(\vec{a},\vec{\san{x}},\vec{\san{i}}) \rangle = 0$.
Thus, the right hand side of Eq.~(\ref{eq-mixture-of-eve-state})
is an eigen value decomposition. Moreover, if $\vec{a} + \vec{b} \in C$,
then we have $\ket{\vartheta(\vec{a},\vec{\san{x}},\vec{\san{j}})} =
\ket{\vartheta(\vec{b},\vec{\san{x}},\vec{\san{j}})}$.
\end{remark}
\begin{proof}
For any $\vec{x} \in C$ and $\vec{a} \in \mathbb{F}_2^m$, we can rewrite
\begin{eqnarray*}
\ket{\varphi(\vec{x}+\vec{a},\vec{\san{x}})} &=&
\frac{1}{\sqrt{P_{\san{X}}^m(\vec{\san{x}})}} \sum_{\vec{\san{j}} \in \san{J}}
\sum_{\vec{\san{c}} \in C^\bot} 
(-1)^{(\vec{x}+\vec{a})\cdot(\vec{\san{j}}+\vec{\san{c}})} \\
&& ~~~\sqrt{ P_{\san{XZ}}^m(\vec{\san{x}}, \vec{\san{j}} + \vec{\san{c}} )}
\ket{\vec{\san{x}}, \vec{\san{j}} + \vec{\san{c}} } \\
&=& \sum_{\vec{\san{j}} \in \san{J}}
(-1)^{(\vec{x}+\vec{a})\cdot \vec{\san{j}} }
\sqrt{ P_{\san{J}|\san{X}^m = \vec{\san{x}} }(\vec{\san{j}}) }
\ket{\vartheta(\vec{a}, \vec{\san{x}}, \vec{\san{j}} ) }.
\end{eqnarray*}
Then, we have
\begin{eqnarray*}
\lefteqn{ 
\sum_{\vec{x} \in C} \frac{1}{|C|} \sigma_{E^m}^{\vec{x}+\vec{a}, \vec{\san{x}} } } \\
&=& \sum_{\vec{x} \in C} \frac{1}{|C|} 
\sum_{\vec{\san{i}}, \vec{\san{j}} \in \san{J} } 
(-1)^{(\vec{x}+\vec{a})\cdot(\vec{\san{i}}+\vec{\san{j}}) }
\sqrt{ P_{\san{J}|\san{X}^m = \vec{\san{x}} }(\vec{\san{i}})
P_{\san{J}|\san{X}^m = \vec{\san{x}} }(\vec{\san{j}}) } \\
&& ~~~\ket{\vartheta(\vec{a}, \vec{\san{x}}, \vec{\san{i}} ) }
\bra{\vartheta(\vec{a}, \vec{\san{x}}, \vec{\san{j}} ) } \\
&=& \sum_{\vec{\san{i}}, \vec{\san{j}} \in \san{J} }
(-1)^{\vec{a} \cdot(\vec{\san{i}}+\vec{\san{j}}) }
\sum_{\vec{x} \in C} \frac{1}{|C|}
(-1)^{\vec{x} \cdot(\vec{\san{i}}+\vec{\san{j}}) }
\sqrt{ P_{\san{J}|\san{X}^m = \vec{\san{x}} }(\vec{\san{i}})
P_{\san{J}|\san{X}^m = \vec{\san{x}} }(\vec{\san{j}}) } \\
&& ~~~\ket{\vartheta(\vec{a}, \vec{\san{x}}, \vec{\san{i}} ) }
\bra{\vartheta(\vec{a}, \vec{\san{x}}, \vec{\san{j}} ) } \\
&=& \sum_{\vec{\san{j}} \in \san{J} }
P_{\san{J}|\san{X}^m = \vec{\san{x}} }(\vec{\san{j}}) 
\ket{\vartheta(\vec{a}, \vec{\san{x}}, \vec{\san{j}} )}
\bra{\vartheta(\vec{a}, \vec{\san{x}}, \vec{\san{j}} )},
\end{eqnarray*}
where $\cdot$ is the standard inner product on the vector space
 $\mathbb{F}_2^m$, and 
we used the following equality,
\begin{eqnarray*}
\sum_{\vec{x} \in C} (-1)^{\vec{x} \cdot (\vec{\san{i}} + \vec{\san{j}} ) }
= 0
\end{eqnarray*}
for $\vec{\san{i}} \neq \vec{\san{j}}$.
\end{proof}

\section{Bell diagonal state is the worst case}
\label{bell-worst}

In this section, we show that the evaluation 
of the key rate formula for a Bell diagonal state
is the worst case. 
Let $\sigma_{AB}$ be a two-qubit density operator
such that Bell diagonal entries are
$\{ P_{\san{XZ}}(\san{x},\san{z})\}$, i.e.,
$\bra{\psi(\san{x},\san{z})} \sigma_{AB} \ket{\psi(\san{x},\san{z})} =
P_{\san{XZ}}(\san{x},\san{z})$.
Let $\{ \san{XZ}(\san{x},\san{z})\}_{\san{x},\san{z} \in \mathbb{F}_2}$
be the Pauli operators on the qubit,
let $\sigma_{AB}^{\san{s},\san{t}} := \san{XZ}(\san{s},\san{t})^{\otimes
2} \sigma_{AB} \san{XZ}(\san{s},\san{t})^{\otimes 2}$,
and let 
\begin{eqnarray*}
\hat{\sigma}_{AB} := 
\frac{1}{4} \sum_{\san{s}, \san{t} \in \mathbb{F}_2} 
\sigma_{AB}^{\san{s},\san{t}}
\end{eqnarray*} 
be the discrete-twirled operator of $\sigma_{AB}$.
Note that $\hat{\sigma}_{AB}$ is of the form
$\sum_{\san{x}, \san{z} \in \mathbb{F}_2} P_{\san{XZ}}(\san{x},\san{z})
\ket{\psi(\san{x},\san{z})}\bra{\psi(\san{x},\san{z})}$
\cite{hamada:03}.
Let $\sigma_{U_1 U_2 W_1 E_1 E_2}^{\vec{\san{s}},\vec{\san{t}}}$ be a density
operator derived from a purification 
$\sigma_{A^2B^2E^2}^{\vec{\san{s}},\vec{\san{t}}}$ of 
$\sigma_{A^2B^2}^{\vec{\san{s}},\vec{\san{t}}}$ in the same way as
$\sigma_{U_1 U_2 W_1 E_1 E_2}$ is derived from 
a purification $\sigma_{ABE}^{\otimes 2}$ of $\sigma_{AB}^{\otimes 2}$ 
in Section 2,
where $\sigma_{A^2B^2}^{\vec{\san{s}}, \vec{\san{t}}} :=
\sigma_{AB}^{\san{s}_1,\san{t}_1} \otimes \sigma_{AB}^{\san{s}_2,
\san{t}_2}$
for $\vec{\san{s}} = (\san{s}_1,\san{s}_2)$ and
$\vec{\san{t}} = (\san{t}_1,\san{t}_2)$.
Since a phase flip error does not affect the measurement by 
$\{ \ket{0_z}, \ket{1_z}\}$-basis,
and since a bit flip error only permutate the indices of
measurement results, we have
\begin{eqnarray*}
H_{\sigma}(U_1 U_2 | W_1 E_1 E_2) = 
H_{\sigma^{\vec{\san{s}}, \vec{\san{t}}}}(U_1 U_2 | W_1 E_1 E_2).
\end{eqnarray*}
Let 
\begin{eqnarray*}
\ket{\Phi_{A B E \san{S} \san{T} \san{S}^\prime \san{T}^\prime}}
:= \sum_{\san{s},\san{t} \in \mathbb{F}_2} \frac{1}{2}
\ket{\Phi_{ABE}^{\san{s},\san{t}}} \ket{\san{s},\san{t},\san{s},\san{t}}
\end{eqnarray*}
be a purification of $\tilde{\sigma}_{AB}$,
where
$\ket{\Phi_{ABE}^{\san{s},\san{t}}}\bra{\Phi_{ABE}^{\san{s},\san{t}}} := 
\sigma_{ABE}^{\san{s},\san{t}}$.
Let $\tilde{\sigma}_{U_1 U_2 W_1 E_1 E_2 \vec{\san{S}} \vec{\san{T}}
\vec{\san{S}^\prime} \vec{\san{T}}}$ be a density operator derived
from $\Phi_{ABE\san{S}\san{T} \san{S}^\prime \san{T}^\prime}^{\otimes
2}$
in the same way as $\sigma_{U_1 U_2 W_1 E_1 E_2}$ is derived from
$\sigma_{ABE}^{\otimes 2}$.
Then, by using the strong subadditivity of von
Neumann entropy, 
we have
\begin{eqnarray*}
\lefteqn{
H_{\tilde{\sigma}}(U_1 U_2 | W_1 E_1 E_2 \vec{\san{S}} \vec{\san{T}}
\vec{\san{S}^\prime} \vec{\san{T}^\prime})
} \\
&\le&
H_{\tilde{\sigma}}(U_1 U_2 | W_1 E_1 E_2 \vec{\san{S}} \vec{\san{T}}) \\
&=& \sum_{\vec{s}, \vec{t} \in \mathbb{F}_2^2}
\frac{1}{16} H_{\sigma^{\vec{\san{s}}, \vec{\san{t}}}}(
U_1 U_2 | W_1 E_1 E_2) \\
&=& H_{\sigma}(U_1 U_2 | W_1 E_1 E_2).
\end{eqnarray*}
In the similar manner, we have
\begin{eqnarray*}
\lefteqn{
H_{\tilde{\sigma}}(U_2 | W_1 U_1 E_1 E_2 \vec{\san{S}} \vec{\san{T}}
\vec{\san{S}^\prime} \vec{\san{T}^\prime})
} \\
&\le&
H_{\sigma}(U_2 | W_1 U_1 E_1 E_2).
\end{eqnarray*}
On the other hand, $P_{W_1}$ and 
$P_{W_2|W_1=0}$ are invariant under the
discrete twirling operation.  
Thus, Bell diagonal state is the worst case
for a fixed Bell diagonal entries 
$\{ P_{\san{XZ}}(\san{x},\san{z})\}$.


\end{document}